\documentclass[11pt]{article}
\usepackage[utf8]{inputenc}

\usepackage{tkz-graph}

\usepackage{cancel}
\usepackage[normalem]{ulem}
\usepackage[T1]{fontenc} 
\usepackage{color}
\usepackage[leqno]{amsmath}
\usepackage{amsfonts, amsthm, amssymb, commath} 
\usepackage{mathtools}
\usepackage{graphicx}
\usepackage{mathrsfs}
\usepackage{caption}
\usepackage{subcaption}
\usepackage{listings}
\usepackage[]{eufrak}
\usepackage[]{hyperref}
\usepackage[]{float}
\usepackage{tikz}
\usetikzlibrary{calc, arrows,matrix,positioning,scopes, intersections}
\usetikzlibrary{shapes.geometric}
\usetikzlibrary{decorations.markings}
\usepackage{tikz-cd}
\usetikzlibrary{cd}
\usepackage{booktabs}
\usepackage{enumerate}
\usepackage{eufrak}
\usepackage{MnSymbol}
\usepackage{fancyhdr}
\usepackage{tikz}
\usepackage{algorithm}
\usepackage{algorithmic}
\usetikzlibrary{matrix,decorations.pathreplacing,calc}
\pgfkeys{tikz/mymatrixenv/.style={decoration=brace,every left delimiter/.style={xshift=3pt},every right delimiter/.style={xshift=-3pt}}}
\pgfkeys{tikz/mymatrix/.style={matrix of math nodes,left delimiter={(}, right delimiter={)}, inner sep=1pt,column sep=1.5em,row sep=0.5em,nodes={inner sep=0pt}}}
\pgfkeys{tikz/mymatrixbrace/.style={decorate,thick}}

\usepackage{cuted}
\newcommand{\Tr}{\mathop{\bf Tr}}
\newcommand{\ca}[1]{\mathcal{#1}}

\newcommand{\bb}[1]{\mathbb{#1}}

\newcommand{\rank}{\text{\bf rank}}

\newcommand{\vect}{\mathop{\bf {vec}}}

\newcommand{\dom}{\mathop{\bf dom}} 

\makeatletter
\newcommand{\leqnomode}{\tagsleft@true}
\newcommand{\reqnomode}{\tagsleft@false}
\makeatother
\renewcommand{\theenumi}{\alph{enumi}}

\newtheorem*{remark}{Remark}
\newtheorem{theorem}{Theorem}[section]
\newtheorem{corollary}{Corollary}[theorem]
\newtheorem{lemma}[theorem]{Lemma}
\newtheorem{proposition}[theorem]{Proposition}

\tikzset{%
    add/.style args={#1 and #2}{
        to path={%
 ($(\tikztostart)!-#1!(\tikztotarget)$)--($(\tikztotarget)!-#2!(\tikztostart)$)%
  \tikztonodes},add/.default={.2 and .2}}
}  
\usepackage{appendix}
\usepackage{cite}
\usepackage{balance}
\usepackage[margin=1.0in]{geometry}
\begin{document}
\title{\bf Policy Gradient-based Algorithms for Continuous-time Linear Quadratic Control}
\date{June 12, 2020}

\author{Jingjing Bu, Afshin Mesbahi, and Mehran Mesbahi\thanks{The authors are with the University of Washington; Emails: {\em bu+amesbahi+mesbahi@uw.edu}}}

\maketitle

\begin{abstract}
  We consider the continuous-time Linear-Quadratic-Regulator (LQR) problem in terms of optimizing a real-valued matrix function over the set of feedback gains. The results developed are in parallel to those in Bu {\em et al.}~\cite{bu2019lqr} for discrete-time LTI systems. In this direction, we characterize several analytical properties (smoothness, coerciveness, quadratic growth) that are crucial in the analysis of gradient-based algorithms. We also point out similarities and distinctive features of the
  continuous time setup in comparison with its discrete time analog.
 First, we examine three types of well-posed flows direct policy update for LQR: gradient flow, natural gradient flow and the quasi-Newton flow. The coercive property of the corresponding cost function suggests that these flows admit unique solutions while the gradient dominated property indicates that the underling Lyapunov functionals decay at an exponential rate; quadratic growth on the other hand guarantees that the trajectories of these flows are exponentially stable in the sense of Lyapunov. We then discuss the forward Euler discretization of these flows, realized as gradient descent, natural gradient descent and quasi-Newton iteration. We present stepsize criteria for gradient descent and natural gradient descent, guaranteeing that both algorithms converge linearly to the global optima. An optimal stepsize for the quasi-Newton iteration is also proposed, guaranteeing a $Q$-quadratic convergence rate--and in the meantime--recovering the Kleinman-Newton iteration.
  Lastly, we examine LQR state feedback synthesis with a sparsity pattern. In this case, we develop the necessary formalism and insights for projected gradient descent, allowing us to guarantee a sublinear rate of convergence to a first-order stationary point.
\end{abstract}

\section{Introduction}
\label{sec:intro}
Linear-Quadratic-Regulator (LQR) has historically been one of the pillars of control theory.
It is formulated around an optimization problem to determine control inputs to a linear dynamical system in order to minimize a given (integral) quadratic cost function over an infinite horizon.
From practical point of view, a fundamental property of LQR controller design is that the resulting optimal input is in the form of state feedback; as such the optimal control input can be represented as a constant feedback gain that acts on the state of the system~\cite{Kalman_BSMM_1960,Anderson_book_1990}.
The state feedback gain that ``solves'' the the infinite-horizon LQR problem, in turn, can be obtained from solving the algebraic Riccati equation (ARE). That is, in the traditional approach to LQR synthesis, the state feedback gain is only obtained after obtaining the ``certificate'' or the cost-to-go for the underlying optimal control problem.
There is a large number of works on the solution of the ARE including those based on iterative algorithms~\cite{Hewer1971TAC}, algebraic solution methods~\cite{Lancaster1995algebraic}, and semidefinite programming~\cite{Balakrishnan2003TAC}.

Although the cost function plays an important role in the LQR problem, it is generally difficult  to {\em directly} compute the optimal policy without going through the Riccati equation. This approach in the meantime, is in sharp contrast to how one would typically go about minimizing a cost function over the variable of interest, say, through a gradient descent. This is essentially due to the dynamic nature of the control problem where an immediately optimal action may not be optimal over the (infinite) time horizon. Recently, there has been a surge of interest in constructing optimal control strategies directly.
These studies have been partially inspired by the application of
learning algorithms, such as Reinforcement learning (RL), where 
using the principles of Dynamic Programming (DP), one can devise 
real-time model-free methods for optimal control problems for both continuous-time and discrete-time systems~\cite{Jiang2012Auto,Young2012Auto,Lee2019TAC,Bradtke1994ACC,Lewis2009CSM,lewis2012reinforcement,Chun2016IJC}.
For example, it has been shown that the solution of Q-learning policy iteration algorithm for discrete-time systems converges to the optimal solution of the LQR problem~\cite{Bradtke1994ACC}. \par
Meanwhile with the emergence of complex distributed systems, it is becoming increasingly important to design a feedback controllers that abide by certain sparsity patterns mirroring the topology of the system. In this setting, each node in the ``network'' forms its control action by only employing local information from its neighbors.
This structured synthesis problem has gained a lot of attention in the  systems and control community. However, there are a number of mathematical complications in formulating such problems. For example, even the existence of a structured stabilizing gain matrix is nontrivial to assert. This general structured synthesis has a long history in the control literature; we shall only list some of the relevant literature to the present work.\footnote{We are concerned with directly updating the structured control polices.} For example, in~\cite{wenk1980parameter}, a combined primal-dual with penalty function method has been employed to obtain a feedback controller with arbitrary constrained zero pattern. The work~\cite{jilg2013optimized} proposes a relaxed mixed-integer-semdefinite-programming via which the graph topology and sparsity pattern are enforced. In \cite{maartensson2009gradient}, the authors use a projected gradient descent method to solve the structured synthesis problem by projecting the solution onto the graph structure. \par
In this paper, similar to \cite{maartensson2009gradient}, we consider the problem of devising first order algorithms for obtaining the optimal LQR state feedback gain for continuous LTI systems with guaranteed convergence properties.
In this direction, we first adjust the LQR problem formulation in order
to make it independent of the system initial conditions.
In order to eliminate this dependence, we adopt a cost function that sums the traditional LQR cost over a set of linearly independent initial conditions. The new cost function can be viewed as a well-defined matrix function over stabilizing feedback gains. We argue that this formulation (see \S~\ref{sec:cost_function} for details) is necessary for the adoption of direct learning algorithms for LQR-type problems. More importantly, in this setting, we prove that the cost function is smooth, coercive and gradient dominant.
We then proceed to show that the LQR function over the set of stabilizing state feedback gains does attain a minimum as its sublevel sets are compact. 
Then, using the topological and metrical properties of the set of
static stabilizing feedback gains, one can conclude that the optimization problem does attain its global minimum. 

The problem of solving the \emph{discrete-time} LQR using direct policy gradient has recently been addressed in \cite{fazel2018global}, where it is shown that direct policy gradient in fact converges to the optimal feedback gain.\footnote{This setup was also considered in~\cite{mart2012phd}, without convergence analysis.}
In~\cite{fazel2018global}, the gradient dominance property, introduced in~\cite{polyak1963gradient}, is used to guarantee the global convergence of gradient descent method.
The present work focuses on LQR \emph{continuous-time LTI systems}. We first show the cost function for this continuous-time LQR is also \emph{gradient dominant}. We then discuss the autonomous gradient flow over the set of Hurwitz stabilizing controllers. For this gradient flow, we show that there exists a unique trajectory for all time $t$ and every initial condition. We then show that the trajectory is exponentially stable in the Lyapunov sense. We next move on to the discretization of this gradient flow. 
In this direction, the required stepsize for the discretized gradient descent algorithm can be adapted using the coerciveness of the cost function; the stepsize selection
 requires particular attention as one needs to ensure that
 the updated feedback gain remains stabilizing.
As such, both the function value and feedback gains converge linearly to the corresponding global minimums.
For this purpose, we first introduce an expression for the Hessian of the LQR cost function.
Then, an upper bound on this Hessian over the sublevel set determined by the initial condition is computed, which is sufficient to determine an upper bound for the stepsize. By using a stepsize equal to the inverse of the computed bound, we can then conclude the linear rate for the global convergence of gradient descent.
Analogous results natural gradient flow and quasi-Newton flow, as well as their discretized versions
are also provided for continuous LQR.

We then extend the proposed approach to the problem of designing feedback gains with an arbitrary zero pattern. Our setup is inspired by the scheme adopted in~\cite{mart2012phd} to solve \emph{discrete-time} LQR. We propose a formalism to set up the problem appropriately in the context of first
order direct policy updates- as such, projected gradient descent has a simple realization. 
The new optimization problem over the set of structured stabilizing controllers does not necessary 
possess the ``gradient dominance'' property. However, we are able to employ the machinery developed for unstructured LQR for this extension. In particular, 
we first describe the initial-condition independent formulation as necessary for the learning setting; we then show the cost function can be equivalently defined as the unstructured LQR cost function restricted to the linear space defined by the interaction graph $\ca G$; as such, the cost function is smooth in the subspace topology and has a coercive property; the gradient and Hessian can be thus clarified and a natural choice of stepsize can be acquired by bounding the Hessian over the initial sublevel set.\par
We remark that the proofs here are parallel to those presented in~\cite{bu2019lqr}. However, since we are dealing with continuous-time systems, the details are inevitably distinct. For completeness of our presentation, we attempt to provide complete proofs as much as possible- occasionally,
we point out that the proofs closely parallel the discrete-time analogue and
make the appropriate reference.
 \par
The remainder of this paper is structured as follows. The continuous LQR problem statement and the related definitions are provided in~\S~\ref{sec:definition}. The cost function over a set of linearly independent initial conditions is also defined in~\S~\ref{sec:definition}. In \S~\ref{sec:function_properties}, we show the initial condition independent formulation of the LQR cost; the corresponding function is smooth, coercive and has compact sublevel sets. In particular, we prove that the function is gradient dominant. \S~\ref{sec:gf} introduces the autonomous gradient system over the set of stabilizing controllers. We show the global existence and uniqueness of a solution trajectory. Moreover, we prove that the trajectory is exponentially stable in the sense of Lyapunov. In~\S~\ref{sec:discretization_gradient_flow}, we discuss means of discretizing the gradient flow. We show that the Forward Euler method with a suitable choice of stepsize would guarantee the linear convergence of both the Lyapunov functional and the underlying iterates. We also discuss an application to solve structured design by \emph{Projected Gradient Descent}. In \S~\ref{sec:numerical_results}, we present several simulation results to illustrate the proposed results, and finally we conclude the paper in~\S~\ref{sec:discussion}.
\section{Notation and Preliminaries}\label{sec:notations}
We denote by ${\bb M}_{n \times m}(\bb R)$ the set of $n \times m$ real matrices and ${\bb {GL}}_n(\bb R)$ as the set of invertible square matrices; $\bb R^n$ denotes the $n$-dimensional real Euclidean space with $n=1$ identified with real number.
$\bb N$ denotes the set of natural numbers.
$\bb S_n$ denotes the set of $n \times n$ real symmetric matrices.
$A^\top$, $\rho(A)$, $\rank(A)$, $\Tr (A)$, $\vect(A)$, $A \otimes B$, and $\text{rbd } \ca K$ represent the transpose of $A$, the spectral radius of $A$, the rank of $A$, and the trace of $A$, the vectorization of $A$, Kronecker product $A$ and $B$, the relative boundary of the set $\ca K$, respectively.
The real inner product between a pair of vectors $x$ and $y$ is denoted by $\langle x,y \rangle$. $\|A\|_2$ denotes the spectral (operator) norm of a square matrix and $\|A\|_F$ denotes the Frobenius norm.
Lastly, the notation $A \succeq B$ for two symmetric matrices refers to the positive semi-definiteness of the matrix difference $A-B$; analogously for positive definiteness and the notation ``$\succ$''.
\par
We use $C^{\omega}(U)$ to denote the set of real analytic functions over an open set $U \subseteq \bb R^n$. A function $f: U \to \bb R$ is $C^{\infty}$-smooth if it is infinitely differentiable. A function $f$ is $L$-smooth when $f$ is \emph{continuously differentiable} and the gradient is $L$-Lipschitz, i.e., $\|\nabla f(x)-\nabla f(y)\| \le L \|x-y\|$.\par
A pair $(A, B)$ with $A \in {\bb M}_{n \times n}(\bb R)$ and $B \in {\bb M}_{n \times m}( \bb R)$ is controllable if it satisfies Kalman rank condition~\cite{sontag2013mathematical}
\begin{align*}
  \rank([B, AB, A^2B, \dots, A^{n-1}B]) = n.
\end{align*}
We will frequently use several linear algebra facts on matrix equations; some of these facts are collected in the following proposition.
\begin{proposition}
  \label{prop:linalg_facts}
  The following relations hold:
  \begin{enumerate}
    \item $\vect(ABC) = (C^{\top} \otimes A) \vect(B)$.
    \item When $X \succ 0$,
\begin{align}
  \label{eq:psd_ineq1}
 M^{\top} X N + N^{\top} X M \succeq - ( a M^{\top} X M + \frac{1}{a}N^{\top} X N),\\
  \label{eq:psd_ineq2}
 M^{\top} X N + N^{\top} X M \preceq a M^{\top} X M +  \frac{1}{a} N^{\top} X N,
\end{align}
where $M, N \in {\bb M}_{n \times m}(\bb R)$ with $m \le n$ and $a \in \bb R_+$.
\item Suppose $A \in {\bb M}_{n \times n}(\bb R)$ is \emph{Hurwitz stable}, i.e., $\max_i \text{Re}(\lambda_i(A)) < 0$. Then
\begin{align*}
  A^{\top} X + XA + Q =0
\end{align*}
has a unique solution. If $Q \succ 0$, then $X \succ 0$. Moreover, if $Y$ satisfies
\begin{align*}
  A^{\top} Y + Y A + O = 0,
\end{align*}
with $O \preceq Q$, then $Y \preceq X$.
  \end{enumerate}
\end{proposition}
The proofs of these observations can be found in~\cite{horn2012matrix}.
\section{Problem Setup and its Analytical Properties}
\label{sec:definition}
In this section, we provide an overview of continuous-time LQR and in particular, its modified initial-condition independent version, as well as analytic observations that we believe are of independent interest. Although the reader might know of the extensive LQR literature, we note that some of these observations have only become necessary when the LQR problem is viewed {\em directly} in terms of optimizing an integral cost function over the set of stabilizing feedback gains. Throughout this section and section $4,5$, we shall focus on LQR with standard assumptions, i.e., $(A, B)$ is stabilizable, $Q = Q^{\top} \succeq 0$,  $R = R^{\top} \succ 0$ and eigenvalues of $A$ on the imaginary axis is $(Q, A)$-observale. 
%
\subsection{Continuous Linear-Quadratic-Regulator Problem} 
In the standard setup of linear optimal control problems, we consider a continuous linear-time-invariant system,
\begin{align*}
  \dot{x}(t) = Ax(t) + B u(t),
\end{align*}
where $A \in {\bb M}_{n \times n}( \bb R)$, $B \in  {\bb M}_{n \times m}( \bb R)$ and $(A, B)$ is stabilizable. The Linear-Quadratic-Regulator problem is to devise a linear feedback controller $K \in {\bb M}_{m \times n} (\bb R)$, s.t., with $u(t) = -K x(t)$, in order to minimize the following cost function,
\begin{align*}
  J(x_0) = \int_{0}^{\infty} \left[ \langle x(t), Q x(t)\rangle + \langle u(t), R u(t)\rangle\right] dt,
\end{align*}
where $x_0$ is the initial condition, $Q = Q^{\top} \succeq 0$, $R = R^{\top} \succ 0$. 
This problem is traditionally solved via the principles of dynamic programming, leading to the Algebraic Riccati Equation (ARE)~\cite{Anderson_book_1990}. 
\subsection{Cost function for direct policy update}\label{sec:cost_function}
In order to update the control policy directly, it will be conceptually appealing to consider the cost function as a matrix function over the feedback gains. We may naively define $J \colon {\bb M}_{m \times n}(\bb R) \to \bb R$ by,
\begin{align}
  \label{eq:naive_cost_function}
  K \mapsto J_{x_0}( K) &= \int_{0}^{\infty}\left[ \langle x(t), Q x(t)\rangle + \langle u(t), R u(t)\rangle\right] dt
 \end{align}
for some fixed (nonzero) initial condition $x_0 \in \bb R^n$.
Our first task in this direct optimization setup is determining the domain over which the function is well-defined. 
In other words, we are interested in effective domain of the function $\dom(J) = \{K \in {\bb M}_{m \times n} (\bb R): J_{x_0}(K) < +\infty\}$. The answer to this seemingly natural analytical question turns out to be subtle. If $K$ is Hurwitz, i.e., all eigenvalues of $A-BK$ have negative real parts, then $K \in \dom(J)$. On the other hand, for a non-stabilizing $K$, i.e., $max_i\text{Re}(\lambda_i)(A-BK) \ge 0$, when the system matrix $A-B K$ has both stable and unstable modes, if $x_0$ is chosen to be in the span of eigenspace corresponding to stable modes, $J(K) < \infty$. Namely, $\{K: A-BK \text{ is Hurwitz}\}$ is a proper subset of $\dom(J)$. Indeed, $\{K: A-BK \text{ is Hurwitz}\}$ is the interior of $\dom(J)$.
\begin{lemma}
  Suppose that the nonzero $x_0 \in \bb R^n$ is fixed. If $J$ is defined by~\eqref{eq:naive_cost_function}, then the set $\ca H \coloneqq \{K \in {\bb M}_{m \times n} (\bb R): A-BK \text{ is Hurwitz}\}$ is the interior of $\dom(J)$.
\end{lemma}
\begin{proof}
  Clearly $\ca H \subseteq \text{int}(\dom(J))$. On the other hand,
  for any $M \in \dom(J)\setminus \ca H$ and every $\varepsilon > 0$, by Proposition $3.1$ in~\cite{bu2019lqr_calculus}, there is some $N \in \bb M_{m \times n}(\bb R)$ such that $\|M-N\|_F < \varepsilon$ and the projection of $x_0$ onto every eigenvector of $A-BN$ is nontrivial. We observe that $\|A-BM - (A-BN)\|_F \le \|B\|_F\|M-N\|_F$. Since $\max \circ {\bf Re}$ is continuous and $\max \circ {\bf Re}(A-BM) \ge 0$, $\rho(A-BN) \ge 0$. As such, $J_{x_0}(N) = \infty$ and $N \notin \dom(J)$. So $M \notin \text{int}(\dom(J))$. Consequently, $\text{int}(\dom(J)) = \ca H$.
\end{proof}
The lemma essentially implies that $J_{x_0}(K)$ is not differentiable everywhere on its domain. Indeed,
although $J_{x_0}(K)$ is differentiable in $\ca H$, it is not differentiable on $\dom(J) \setminus \ca H$. This complication is rather unnecessary as we are primarily interested in stabilizing controllers. This motivates us to examine initial condition independent formulation for the LQR synthesis.\footnote{Indeed this is necessary to formulate an unconstrained optimization problem to implicitly deal with stability.}
\subsection{Initial condition independent formulation of LQR}
Ideally, the objective function $f: {\bb M}_{m \times n} ( \bb R) \to \bb R$ for our LQR 
formulation has an effective domain that coincides with the set of stabilizing feedback gains $\ca H$. This can be achieved by choosing a spanning set $\{x_0^1, \dots, x_0^n\} \subseteq \bb R^n$ and defining,\footnote{Of course, we may choose the standard basis $\{e_1, \dots, e_n\}$. The choice of arbitrary spanning set is to retain generality.}
\begin{align*}
  f(K) = \sum_{j=1}^n J_{x_0^j}(K).
\end{align*}
As such, the function $f$ would be infinite if $K$ is not stabilizing (see Lemma \ref{lemma:coercive} for details).
The initial condition independent LQR cost $f$ now enjoys several favorable properties, {e.g.}, $f$ is differentiable over its entire effective domain and diverges to infinity when $K$ tends to the boundary of the effective domain or when $K$ diverges to infinity, i.e., $f$ is coercive. More importantly, for every $K \in \dom(f)$, $f(K)$ can be written as,
\begin{align}
  \label{f(k)}
  f(K) = \sum_{j=1}^n \Tr(X X_0^j),
\end{align}
where $X_0^j = x_0^j (x_0^j)^\top$ and $X$ satisfies the Lyapunov matrix equation,
\begin{align}
  \label{eq:lyapunov_matrix}
(A-BK)^\top X + X(A-BK) + K^{\top} R K + Q = 0.
  \end{align}
Note that from a purely matrix function perspective, $J_{x_0^j}(K)$ does not necessarily admit the compact form $J_{x_0^j}(K)= \Tr(X X_0^j)$ for every $K \in \dom(J_{x_0^j}(K))$. As explained in \S \ref{sec:cost_function}, this is due to the fact that $X$ only makes sense when if $K$ is stabilizing but $\dom(J_{x_0^j}(K))$ contains non-stabilizing feedback gains.
\par
\begin{remark}
  Alternatively, we could let $x_0 \sim \ca D$ where $\ca D$ denotes some probability distribution, and let
  \begin{align*}
    f(K) = \bb E_{x_0 \sim \ca D}(J(x_0)).
  \end{align*}
  As long as the samples span the whole space with probability $1$, the function enjoys same properties as we have defined above. This is indeed the formulation adopted in~\cite{mart2012phd, fazel2018global}\footnote{It is discrete LQR investigated in these work.}.
\end{remark}
\subsection{Analytical Properties of the LQR cost function}
\label{sec:function_properties}
In this section, we review some of the properties of $f(K)$. In particular, we show that,
\begin{itemize}
  \item $f$ is a real analytic function over its domain.
  \item $f$ is coercive and has compact sublevel sets.
    \item $f$ is gradient dominated over all of its sublevel sets.
\end{itemize}
To simplify the notation, throughout the paper we shall denote,
$$A_K \coloneqq A-BK, \qquad {\bf \Sigma} = \sum_{j=1}^n x_0^j (x_0^j)^{\top}, \qquad {\bf N}_K=RK-B^{\top} X.$$
To begin with, recall that $\ca H$ denotes the set of Hurwitiz stabilizing feedback gains $\ca H =\{K \in {\bb M}_{m \times n}(\bb R): A-BK \text{ is Hurwitz}\}$. It is known that $\ca H$ is open, contractible and unbounded~\cite{bu2019topological_mimo}.
We next observe that $f$ is real analytic over $\ca H$.
\begin{proposition}
  One has
  $f \in C^{\omega}(\ca H)$.
\end{proposition}
\begin{proof}
  We note the map $K \mapsto X(K)$ where $X(K)$ is the solution to the Lyapunov equation~\eqref{eq:lyapunov_matrix}
is $C^{\omega}$ since
\begin{align*}
  \vect(X) = \left( I \otimes A_K^{\top} + A_K^{\top} \otimes I \right)^{-1} \vect(K^{\top} R K + Q)
\end{align*}
and by Cramer's Rules, the entries of $X$ are rational functions in the entries of $K$. Moreover,
  $f$ can be viewed in terms of the following composition:
  \begin{align*}
    K \mapsto X(K) \mapsto \Tr(X{\bf \Sigma});
  \end{align*}
  hence $f \in C^{\omega}(\ca H)$.
\end{proof}
We next observe that $f$ is coercive\footnote{We adopt the convention in optimization: $f$ is coercive if $\lim_{\|x\| \to \infty} f(x) = +\infty$. This property in control literature sometimes is referred as \emph{weakly coercive}~\cite{sontag2013mathematical}.}.
\begin{lemma}
  \label{lemma:coercive}
  The function $f$ (\ref{f(k)}) is coercive, i.e.,
  \begin{align*}{}
    \lim_{K_j \to K \in \ca H} f(K_j) \to \infty \text{ and } \lim_{\|K\| \to \infty} f(K) \to \infty.
  \end{align*}
\end{lemma}
\begin{proof}
  Suppose that $\{K_j\}$ is a sequence in $\ca H$ converging to $K \in \partial \ca H$. Denote the sequence $\{X_j\} \subseteq \bb S_n^{++}$ to be the corresponding sequence of value matrices. We claim that the sequence diverges to infinity in the $2$-norm. Namely, $\|X_j\|_2 \to +\infty$ as $j \to \infty$. To show this, it suffices to show that the sequence contains no bounded subsequence. We prove by contradiction. Suppose not, i.e., there exists some bounded subsequence $\{X_{n_k}\}$; then by Weirestrass-Balzano~\cite{rudin1964principles}, there exists some subsubsequence $\{n_{k_j}\}$ such that $X_{n_{k_j}} \to X$ for some $X \succeq 0$. By continuity, $X \succeq 0$ solves the Lyapunov equation,
  \begin{align*}
    A_K^{\top}X + X A_K + Q +K^{\top}R K = 0.
    \end{align*}
    But this is a contradiction: if $(\lambda, v)$ is an eigen pair of $A_K$ with $\lambda = i \beta$ for some $\beta \in \bb R$, then we have
    \begin{align*}
      v^{\top}\left( A_K^{\top}X  + XA_K \right)v + v^{\top}(Q + K^{\top}RK)v = 0,
      \end{align*}
which implies that $Q v = 0$, $Kv = 0$ and $Av = \lambda v$. This is a contradiction to the $(Q, A)$ observability of $\lambda$. Hence, $\{X_j\}$ must be unbounded. It thus follows that $\Tr(X_j {\bf \Sigma}) \to + \infty$ as $j \to +\infty$.\\
  On the other hand, we have
  \begin{align*}
    f(K) &= \Tr({\bf \Sigma}X) \ge \lambda_{\min}({\bf \Sigma})\Tr(X) \ge \lambda_{\min}({\bf \Sigma}) \Tr \left( \int_0^{\infty} e^{A_{K}^{\top} t}(K^{\top} R K + Q) e^{A_{K} t} dt\right) \\
    &\ge \lambda_{\min}({\bf \Sigma}) \Tr(K^{\top}RK + Q) \lambda_{\min}\left( \int_0^{\infty} e^{A_K^{\top} t} e^{A_K t} dt\right).
  \end{align*}
  We note that $G \coloneqq \int_{0}^{\infty}e^{A_K^{\top} t} e^{A_K t} dt$ is the solution of the Lyapunov  equation,
  \begin{align*}
    A_K^{\top}G + GA_K + I = 0.
  \end{align*}
  We observe that for the unit eigenvector $v$ with $Gv = \lambda_{\min}(G) v$ we have,
  \begin{align*}
     v^{\top}(A_K^{\top}G + GA)v = -1 \implies  \lambda_{\min}(G) v^{\top}(A_K^{\top} + A_K) v = -1 \implies  \lambda_{\min}(G) = \frac{1}{-v^{\top}(A_K^{\top} + A_K) v} ,
  \end{align*}
  and 
  \begin{align*}
    \lambda_{\min}(A_K^{\top} + A_K) = \min_{\|u\|_2=1} u^{\top}(A+A^{\top})u \le v^{\top} (A+A^{\top}) v.
  \end{align*}
 Hence $\lambda_{\min}(A_K+A_K^{\top}) < 0$ and it follows that,
  \begin{align*}
    \lambda_{\min}(G) \ge \frac{1}{-\lambda_{\min}(A_K+A_K^{\top})} \ge \frac{1}{\|A_K+A_K^{\top}\|_F} \ge \frac{1}{2\|A\|_F+2\|B\|_F \|K\|_F}.
  \end{align*}
 Thereby,
  \begin{align*}
    f(K) \ge \frac{\lambda_{\min}({\bf \Sigma})( \lambda_{\min}(R) \|K\|_F^2 + \Tr(Q)) }{2\|A\|_F+2\|B\|_F\|K\|_F}.
  \end{align*}
  This now implies that $f(K) \to +\infty$ as $\|K\|_F \to \infty$.
\end{proof}
As a consequence, all sublevel sets of $f(K)$ are compact.
\begin{corollary}
  \label{cor:compact_sublevel_sets}
  For every $\alpha \in \bb R$, the sublevel set 
$$S_{\alpha}(f) \coloneqq \{K \in \ca H: f(K)\le \alpha\}$$ is compact.
\end{corollary}
\begin{proof}
  The proof is identical to the one of Corollary $3.7.1$ in~\cite{bu2019lqr} and thus omitted.
\end{proof}
As $f \in C^{\omega}(\ca H)$, we can now characterize the gradient and Hessian of $f$.
\begin{proposition}
  \label{prop:gradient_hessian}
  For $K \in \ca H$, the gradient of $f(K)$ is given by
\begin{align*}
  \nabla f(K) = 2(RK-B^{\top}X)Y,
\end{align*}
where
\begin{align*}
  Y = \int_0^{\infty} e^{A_Kt} {\bf \Sigma} e^{A_K^{\top}t}dt.
\end{align*}
Moreover, the action of Hessian $\nabla^2 f(K)$ is given by
\begin{align*}
  \nabla^2 f(K)[E, E] = 2\Tr(E^{\top}REY) - 4\Tr(E^{T}B^{\top}(X'(K)[E])Y),
\end{align*}
for every $E \in \bb M_{m \times n}(\bb R)$, where $X'(K)$ is the differential of the map $K \mapsto X(K)$ and $X'(K)[E]$ denotes the action by matrix multiplication.\footnote{To be more precise, the differential of $X: \bb M_{m \times n}(\bb R) \to \bb M_n(\bb R)$ is a map $X': \bb M_{m \times n}(\bb R) \to \ca L(\bb M_{m \times n}(\bb R), \ca L(\bb M_{m \times n}(\bb R), \bb M_n(\bb R)))$ where $\ca L$ denotes bounded linear maps. As such, $X'(K) \in \ca L(\bb M_{m \times n}(\bb R), \bb M_n(\bb R))$ and $X'(K)[E] \in \bb M_n(\bb R)$.}
\end{proposition}
The gradient formula can be found in~\cite{levine1970determination, knapp1972optimal, levine1971optimal, kleinman1968design}. As we shall derive the action of the Hessian, we provide a derivation of gradient formula as well.
\begin{proof}
  Note that $f$ is the composition of $K \mapsto X(K) \mapsto \Tr(X{\bf \Sigma})$.
  Observe the differential $X'(K)$ necessarily satisfies,
\begin{align*}
  -(BE)^{\top} X - X(BE) + A_K^{\top}X'(K)[E] + X'(K)[E] A_K + E^{\top} RK + K^{\top} R E = 0.
\end{align*}
 It thus follows that $X'(K)[E]$ is uniquely defined and,
\begin{align}
  \label{eq:X_derivative}
  X'(K)[E] &= \int_{0}^{\infty} e^{A_K^{\top} t}\left( -(BE)^{\top} X - X(BE) + E^{\top} RK + K^{\top} R E \right) e^{A_K t} dt.
\end{align}
By the chain rule, we have,
\begin{align*}
  \nabla f(K)[E] &= \Tr(X'(K)[E] {\bf \Sigma}) = 2\langle E, (RK-B^{\top}X) \int_0^{\infty} e^{A_K t} {\bf \Sigma} e^{A_K^{\top} t} dt\rangle.
\end{align*}
Hence, $\nabla f(K) = 2(RK-B^{\top} X)Y$. \\
We also note that the differential $Y'(K)$ necessarily satisfies,
\begin{align*}
  (-BE) Y -Y(BE)^{\top} + A_K Y'(K)[E] + Y'(K)[E] A_K^{\top} = 0.
\end{align*}
Hence by the product rule, we have,
\begin{align*}
  \nabla f(K)[E] = 2(RE-B^{\top}X'(K)[E])Y + 2(RK-B^{\top} X)Y'(K)[E].
\end{align*}
It thus follows that,
\begin{align*}
  \nabla^2 f(K)[E, E] &= \langle E, 2(RE-B^{\top} X'(K)[E])Y  \rangle + 2\langle E, (RK-B^{\top}X)Y'(K)[E]\rangle\\
  &= 2\Tr(E^{\top}REY) - 4\Tr(E^{T}B^{\top}X'(K)[E]Y).
\end{align*}
\end{proof}
\begin{remark}
  We may observe the global minimum $K_*$ has non-degenerate Hessian since $X'(K_*)[E]=0$ and $\nabla f(K_*)[E, E]= 2\Tr(E^{\top} R E Y)$, suggesting that $\nabla f(K_*)$ is positive definite.
\end{remark}
Next, observe some ``growth properties'' of $f$~(\ref{f(k)}). More specifically, we show that $f$ is \emph{gradient dominated}~\cite{polyak1963gradient} and has \emph{quadratic growth} over its sublevel sets. Indeed, we shall bound $f(K)-f(K_*)$ in terms of $\Tr({\bf N}_K^{\top} {\bf N}_K)$, a relation
that will be used subsequently. Gradient dominated property is a simple corollary of this fact by noting $\nabla f(K)= {\bf N}_K Y(K)$.
\begin{lemma}
  \label{lemma:natural_gradient_dominated}
  For every $K \in \ca H$, we have
\begin{align*}
  \lambda_1(Y) \lambda_1(R) \|K-K_*\|_F^2 \le f(K) - f(K_*) \le \frac{\|Y_*\|}{\lambda_1(R)} \Tr({\bf N}_K^{\top} {\bf N}_K),
\end{align*}
where $Y_*$ solves $A_{K_*} Y_* + Y_* A_{K_*}^{\top} + {\bf \Sigma} = 0$ and $Y$ solves $A_K Y + Y A_K^{\top} + {\bf \Sigma} = 0$.
\end{lemma}
\begin{proof}
  We note that,
\begin{align*}
  f(K)- f(K_*) = \Tr\left( (X-X_*){\bf \Sigma}\right).
\end{align*}
Our task is essentially to estimate $X-X_*$. Recall that for $K \in \ca H$, $X(K)$ satisfies the equation,
\begin{align}
  \label{eq:myeq2}
  (A-BK)^{\top}X + X(A-BK) + K^{\top} R K + Q = 0,
\end{align}
and $X_*$ solves,
\begin{align}
  \label{eq:myeq1}
  (A-BK_*)^{\top}X_* + X_*(A-BK_*) + K_*^{\top} R K_* + Q = 0.
\end{align}
Taking the difference of~\eqref{eq:myeq2} and~\eqref{eq:myeq1} yields,
\begin{align*}
  &A_K^{\top} X + X{A_K} - {A}_{K_*}^{\top} X_* - X_*{A}_{K_*} + K^{\top} R K - K_*^{\top} RK^{*} = 0 .
    \end{align*}
    A few algebraic operations of above equation yield (recall ${\bf N}_K = RK -B^{\top} X$),
\begin{equation}
  \label{eq:gradient_dominance_difference}
  {A}_{K_*}^{\top} (X-X_*) + (X-X_*){A}_{K_*} + (K-K_*)^{\top}{\bf N}_K + {\bf N}_K^{\top}(K-K_*) -(K-K_*)^{\top}R(K-K_*) = 0,
\end{equation}
By Proposition~\ref{prop:linalg_facts}, for every $\alpha > 0$, we have
\begin{equation*}
  (K-K_*)^{\top}{\bf N}_K + {\bf N}_K^{\top}(K-K_*) \preceq \alpha {\bf N}_K^{\top}{\bf N}_K + \frac{1}{\alpha} (K-K_*)^{\top}(K-K_*).
\end{equation*}
Picking $1/\alpha = \lambda_{1}(R)$, we have
\begin{equation*}
  (K-K_*)^{\top}{\bf N}_K + {\bf N}_K^{\top}(K-K_*) -(K-K_*)^{\top}R(K-K_*) \preceq \frac{1}{\lambda_1(R)} {\bf N}_K^{\top}{\bf N}_K.
\end{equation*}
Let $Z$ be the solution of the Lyapunov equation,
\begin{align*}
  A_{K_*}^{\top}Z+ZA_{K_*} + \frac{1}{\lambda_1(R)}{\bf N}_K^{\top}{\bf N}_K = 0;
\end{align*}
it thus follows that $X-X_* \preceq Z$.
Hence, 
\begin{align*}
  f(K)- f(K_*) &= \Tr((X-X_*){\bf \Sigma}) \le \Tr(Z{\bf \Sigma}) = \frac{1}{\lambda_1(R)}\Tr\left( \int_0^{\infty} e^{A_{K_*}^{\top}t}{\bf N}_K^{\top}{\bf N}_K e^{A_{K_*} t}dt {\bf \Sigma} \right) \\
&\le \frac{1}{\lambda_1(R)}\lambda_{n}\left(\int_0^{\infty} e^{A_{K_*} t} {\bf \Sigma} e^{A_{K_*}t} dt\right)\Tr({\bf N}^{\top}{\bf N}) = \frac{\|Y_*\|}{\lambda_1(R)}\Tr({\bf N}^{\top} {\bf N}).
\end{align*}
For quadratic growth property and following similar steps as in~\eqref{eq:gradient_dominance_difference}, 
we have
\begin{equation}
  \label{eq:quadratic_growth_difference}
  {A}_{K}^{\top} (X-X_*) + (X-X_*){A}_{K} + (K-K_*)^{\top}{\bf N}_{K_*} + {\bf N}_{K_*}^{\top}(K-K_*) +(K-K_*)^{\top}R(K-K_*) = 0.
\end{equation}
But noting that ${\bf N}_{K_*} = RK_* - B^{\top} X_* = 0$, it follows that,
\begin{align*}
  f(K)-f(K_*) = \Tr \left( (K-K_*)^{\top}R(K-K_*) Y\right) \ge \lambda_1(Y) \lambda_1(R) \|K-K_*\|_F^2.
  \end{align*}
  \end{proof}
  We now deduce that over any sublevel set, $f(K)$ is gradient dominated and has quadratic growth at $K_*$.
\begin{corollary}
  \label{cor:gradient_dominance}
  For every $K \ge \ca H$, over the sublevel set $S_{f(K)} = \{K': f(K') \le f(K)\}$, we have
\begin{align*}
  \tau \lambda_1(R) \|K'-K_*\|_F^2 \le f(K') - f(K_*) \le \frac{\|Y_*\|}{4\tau \lambda_1(R)} \langle \nabla f(K'), f(K') \rangle
\end{align*}
for every $K' \in S_{f(K)}$, where $\tau = \min_{K' \in S_{f(K)}} \lambda_1(Y(K'))$.
\end{corollary}
\begin{proof}
  We first note that,
\begin{align*}
    \langle \nabla f(K), \nabla f(K)\rangle = 4 \Tr(YY^{\top} M^{\top}M) \ge 4\lambda_{1}^2(Y)\Tr(M^{\top}M).
\end{align*}
It suffices to lower bound $\lambda_{1}^2(Y)$ over $S_{f(K_0)}$. But note that the map $K \mapsto Y(K) \mapsto \lambda_{1}(Y)$ is continuous; hence $\lambda_{1}(Y)$ achieves minimum over the compact sublevel set and this minimum must be positive since $Y(K')$ is positive definite for every $K' \in S_{f(K)}$. The other inequality is immediate by the definition of $\tau$. 
\end{proof}
\section{Gradient Flows on $\ca H$}
\label{sec:gf}
In this section, we shall examine the autonomous gradient system
\begin{align}
  \label{eq:gradient_flow}
  \dot{K}_t = -\nabla f(K_t).
\end{align}
This is a natural \emph{Ordinary Differential Equation}(ODE) process to minimize the cost $f(K)$ and it can be seen as a continuous limit of gradient descent. The stability properties of gradient flow are tightly related the convergence rate of the forward Euler discretization, namely gradient descent. Elegent Lyapunov-type argument can be employed for proofs and can provide valuable insights on proving the convergence rate of gradient descent. The results presented are in parallel to \S~$4$ in~\cite{bu2019lqr}. \par
We first observe the gradient system~\eqref{eq:gradient_flow} is well-posed.
\begin{lemma}
  For every initial condition $K_0 \in \ca H$, there exists a unique solution for all time $t$, i.e., a solution trajectory $K_t \in C^{\infty}(\bb R_+, \ca H)$ for the initial value problem
\begin{align}
  \label{eq:ivp}
  \begin{cases}
  \dot{K}_t = -\nabla f(K), \\
  K(0) = K_0.
  \end{cases}
\end{align}
Moreover, the trajectory $K_t$ depends smoothly on the initial condition $K_0$.
\end{lemma}
\begin{proof}
  Note $K \mapsto 2(RK-B^{\top}X)Y$ is $C^{\infty}$ smooth.
  The statement then follows from Corollary~\ref{cor:compact_sublevel_sets} and Proposition $3.7$ in~\cite{helmke2012optimization}.
\end{proof}
We now observe the trajectory of the gradient flow are exponentially stable.
\begin{theorem}
  \label{thrm:gf_exponential_trajectory}
For $K_0 \in \ca H$, denote $K_t$ the solution to the IVP~\eqref{eq:ivp}. Then the trajectory $K_t$ is exponentially stable in Lyapunov sense, i.e., 
\begin{align*}
  \|K_t - K_*\|_F^2 \le c e^{\alpha t} \|K_0 - K_*\|_F^2,
\end{align*}
where $c$ and $\alpha$ are constants determined by system parameters $A, B, Q, R$ and initial condition $K_0$.
\end{theorem}
To prove this theorem, we first observe that the energy functional $V(K_t) \coloneqq f(K_t) - f(K_*)$ decays exponentially to $0$.
\begin{lemma}
  \label{lemma:convergence_energy}
  For $K_0 \in \ca H$, denote $K_t$ the solution to the IVP~\eqref{eq:ivp}. Then
\begin{align*}
  f(K_t) - f(K_*) \le e^{\alpha t} (f(K_0)-f(K_*)),
\end{align*}
where $\beta$ is constant determined by system parameters $A,B,Q,R$ and $K_0$.
\end{lemma}
\begin{proof}
  First observe
  \begin{align*}
    \frac{d V(K_t)}{dt} = \langle \nabla f(K_t), \dot{K_t}\rangle = - \|\nabla f(K_t)\|_F^2 \le 0.
  \end{align*}
  So $f(K_t)-f(K_*)$ is monotonically decreasing. As such, the trajectory $\{K_t: t \ge 0\}$ will be completely contained in the sublevel set $S_{f(K_0)}$. Putting $1/\alpha$ to be the constant in Corollary~\ref{cor:gradient_dominance}, i.e., $f(K)-f(K_*) \le (1/\alpha) \langle \nabla f(K), \nabla f(K)\rangle$, we have
\begin{align*}
  \dot{V}(K_t) = -\langle \nabla f(K), \nabla f(K)\rangle \le - \alpha V(K).
\end{align*}
It follows
\begin{align*}
  f(K_t) - f(K_*) \le e^{-\alpha t} (f(K_0)-f(K_*)).
\end{align*}
\end{proof}
\emph{Proof to Theorem~\ref{thrm:exponential_trajectory} (exponentially stable in Lyapunov sense)} 
\begin{proof}
  We first observe the Lyapunov functional is smooth, positive definite and radially unbounded\footnote{In control literature, it sometimes refers as weakly coercive. But this is equivalent to coercive we have proved.}, thus $K_t$ is globally asymptotic stable, i.e.,
  \begin{align*}
    \lim_{t \to \infty} K_t = K_*.
  \end{align*}
  Now note
  \begin{align}
   \label{eq:trajectory_eq1}
    \begin{split}
      \langle K_t - K_*, \dot{K}_t\rangle &= \langle K_t-K_*, -\nabla f(K_t)\rangle = -\Tr((K_t-K_*)^{\top} \nabla f(K_t)) \\
                                       &= -\frac{1}{2} \Tr\left( (K_t-K_*)^{\top}\nabla f(K_t) + (\nabla f(K_t))^{\top}(K_t-K_*)\right).
                                       \end{split}
  \end{align}
  Observe
  \begin{equation*}
    \resizebox{0.98\textwidth}{!}{$
    {\bf U} \coloneqq (K_t-K_*)^{\top} \nabla f(K_t) +[\nabla f(K_t)]^{\top}(K_t-K_*) \preceq \beta[\nabla f(K_t)]^{\top} \nabla f(K_t) + \frac{1}{\beta} (K_t-K_*)^{\top} (K_t-K_*)$}
  \end{equation*}
  for every $\beta > 0$. Further note if we put $\tau = \min_{K \in S_{f(K_0)}} \lambda_1(Y_K)$, by Theorem~\ref{cor:gradient_dominance}, we have
\begin{align*}
  \|K_t - K_*\|_F^2 \le \frac{\|Y_*\|}{\tau^2 \lambda_1^2(R)} \langle \nabla f(K_t), \nabla f(K_t)\rangle.
  \end{align*}
  Thus putting $\beta = \sqrt{\|Y_*\|_2}{\lambda_1^2(R) \tau^2}$, we have
 \begin{align*}
   \Tr(U) \le 2 \sqrt{\frac{\|Y_*\|_2}{\lambda_1^2(R)\tau^2}} \Tr([\nabla f(K_t)]^{\top} \nabla f(K_t)) \eqqcolon c' \Tr([\nabla f(K_t)]^{\top} \nabla f(K_t)).
 \end{align*}
 It follows
 \begin{align}
   \label{eq:trajectory_eq1}
   \langle K_t-K_*, \dot{K}_t\rangle &= -\langle K_t-K_*, \nabla f(K_t)\rangle \ge - \frac{c'}{2} \langle \nabla f(K_t), \nabla f(K_t)\rangle = -\frac{c'}{2} \langle \nabla f(K_t), \dot{K}_t\rangle. 
 \end{align}
Integrating both sides of~\eqref{eq:trajectory_eq1} from $t$ to $\infty$, we have
 \begin{align*}
   &\int_{t}^\infty \langle K_t - K_*, \dot{K}_t\rangle \ge \int_t^{\infty} -\frac{c'}{2} \langle \dot{K}_t, \nabla f(K_t)\rangle \implies  -\frac{1}{2} \|K_t-K_*\|^2 \ge \frac{c'}{2} (- V(K_t)) \\
   &\implies  \|K_t - K_*\|_F^2 \le c' V(K_t) \le c' e^{-\alpha t} (f(K_0)-f(K_*)).
 \end{align*}
 Putting $c = c' (f(K_0)-f(K_*))/(\|K_0-K_*\|_F^2)$ completes the proof.
\end{proof}
\subsection{Discretization of Gradient Flow}
\label{sec:discretization_gradient_flow}
In this section, we concern how to discretize gradient flow and acquire a convergent algorithm that is suitable for computerized practices. As we have observed in Lemma~\ref{lemma:convergence_energy} and Theorem~\ref{thrm:exponential_trajectory}, both the energy functional and the trajectory converge exponentially to the stationary points, ideally we shall acquire a gradient descent algorithm converging linearly for both the function value and the iterates. In this direction, the forward Euler discretization of the {gradient flow} yields,
\begin{align}
  \label{eq:gradient_descent}
  K_{j+1} = K_j - \eta_j \nabla f(K_j),
\end{align}
where $\eta_j$ is a nonnegative stepsize to be determined. The stepsize (or learning rate) should reflect two principles during the iterative process: (1) stay stabilizing and (2) sufficiently decrease the function value. In following, we shall see that the gradient dominated property leads to a stepsize that results in a sufficient decrease in the function values while the coerciveness guarantees that the acquired feedback gain is stabilizing. To begin, we observe that if $K_{j+1} = K_j - \eta_j \nabla f(K_j)$, provided that $K_j$ and $K_{j+1}$ are both stabilizing\footnote{Mind that this is an important assumption as we shall use the solution to Lyapunov matrix equation; the Lyapunov matrix equation is solvable if $A-BK$ is Hurwitz stable.}, the difference of the value matrix $X_{j+1} - X_j$ can be characterized as follows:
\begin{lemma}
  \label{lemma:value_difference}
  If $K_{j+1} = K_j - \eta_j \nabla f(K_j)$ and $K_j, K_{j+1}$ are both stabilizing, then $Z \coloneqq X_{j+1}-X_j$ solves the Lyapunov matrix equation,
  \begin{align*}
  A_{K_{j+1}}^{\top} Z  + Z A_{K_{j+1}}  - 2\eta_j Y_j^{\top} {\bf N}_{j}^{\top} {\bf N}_j - 2\eta_j {\bf N}_j^{\top} {\bf N}_j Y_j  + 4\eta_j^2 Y_j^{\top} {\bf N}_j^{\top} R  {\bf N}_j Y_j = 0.
    \end{align*}
    where we use simplified notation ${\bf N}_j$ and $Y_j$ to denote ${\bf N}_{K_j}$ and $Y_{K_j}$.
  \end{lemma}
  \begin{proof}
    Following the same strategy used in the proof of  Lemma~\ref{lemma:natural_gradient_dominated}, namely, taking the difference of the corresponding Lyapunov matrix equations, we observe that (recall ${\bf N}_K = RK -B^{\top} X$),
      \begin{equation} \label{eq:value_difference}
      \begin{split}
        0 &= {A}_{K_{j+1}}^{\top} (X_{j+1}-X_j) + (X_{j+1} - X_j){A}_{K_{j+1}} + (K_{j+1}-K_{j})^{\top}{\bf N}_{j} + {\bf N}_{j}^{\top}(K_{j+1}-K_{j}) \\
        &\quad + (K_{j+1}-K_{j})^{\top} R (K_{j+1}-K_j).
      \end{split}
   \end{equation}
   Substituting $K_{j+1} - K_j = - 2\eta_j {\bf N}_jY_j$, we then have,
\begin{align*}
  A_{K_{j+1}}^{\top} Z  + Z A_{K_{j+1}}  - 2\eta_j Y_j^{\top} {\bf N}_{j}^{\top} {\bf N}_j - 2\eta_j {\bf N}_j^{\top} {\bf N}_j Y_j  + 4\eta_j^2 Y_j^{\top} {\bf N}_j^{\top} R  {\bf N}_j Y_j = 0.
\end{align*}
   \end{proof}
    We now observe that with appropriately chosen $\eta_j$, we can guarantee a sufficient decrease in the function value while ensuring stabilization.
\begin{lemma}
  \label{lemma:gd_function_decrease}
Consider the sequence $\{K_j\}$ generated by~\eqref{eq:gradient_descent} with stepsize $\eta_j$. 
Denote by $\{X_j\}$ the corresponding Lyapunov matrix solutions with respect to $\{K_j\}$. When $$\eta_j < \sqrt{\frac{1}{c} + \frac{b^2}{4c^2}} - \frac{b}{2c}$$ with
\begin{align*}
  b_j = \frac{f(K_j)\lambda_n(R)}{\lambda_1(Q)} + \frac{4\|B {\bf N}_j Y_j\|_2 f(K)}{\lambda_1(Q) \lambda_1({\bf \Sigma})},\quad c_j = \frac{4 \lambda_1(R)\|B {\bf N}_j Y_j\|_2f(K)}{\lambda_1(Q)\lambda_1({\bf \Sigma})},
\end{align*}
then $\{K_j\}$ is stabilizing for every $j \ge 0$. In particular,
\begin{align*}
  f(K_{j+1}) -f(K_j) < 0.
\end{align*}
\end{lemma}
Before presenting the proof of this result, we shall first outline its basic idea. The crucial property we shall leverage is the compactness of the sublevel sets, analogous to devising the stepsize. If we start at a stabilizing control gain $K$ where the gradient does not vanish and consider the ray of $\{K - \eta \nabla f(K): \eta \ge 0\}$, by compactness of the sublevel set, there is some $\zeta$ for which $f(K') = f(K)$, where $K' \coloneqq K-\zeta \nabla f(K)$ (See Figure~\ref{fig:level_curve}). What we shall demonstrate is that with the stepsize $\eta_j$ given in the Lemma, if $K_{j+1}$ stays in the {compact} sublevel set, then $K_{j+1}$ must stay in the interior of the sublevel set, namely, $f(K_{j+1}) < f(K_j)$. We then proceed to examine two alternatives: (1) $K_{j+1}$ is not stabilizing, or (2) $K_{j+1}$ is stabilizing but $f(K_{j+1}) > f(K_j)$; either alternative would lead to a contradiction.
\begin{figure}[h!]
  \centering
\begin{tikzpicture}[scale=0.6]
    \path[font={\tiny}]
        (0 , 0)   coordinate (A1) 
        (1  , 1)   coordinate (A2)
        (3   , 1)   coordinate (A3)
        (4   , 3)   coordinate (A4)
        (5, 3) coordinate (A5)
        (4, -1) coordinate (A6)
        (2, -1.5) coordinate (A7)
        (-1, 0) coordinate (B1) 
        (2, 2) coordinate (B2)
        (3, 2) coordinate (B3)
        (4, 4) coordinate (B4)
        (6, 4) coordinate (B5)
        (5, -2) coordinate (B6)
        (3, -2.5) coordinate (B7)
        (2, 0) coordinate (C1)
        (2.5, 0.5) coordinate (C2)
        (3, 0.8) coordinate (C3)
        (4, 1.5) coordinate (C4)
        (3.5, 0) coordinate (C5)
        (3, -1) coordinate (C6)
        (2, -1) coordinate (C7)
        ($(B6)!(B7)!(B1)!0.5!(B7)$) coordinate (AA)
        ($(B7)!16.5!(AA)$) coordinate (BB)
    ;
\draw[black, name path=curve 1] plot [smooth cycle] coordinates {(B1) (B2) (B3) (B4) (B5) (B6) (B7)};
\draw[black] plot [smooth cycle] coordinates {(A1) (A2) (A3) (A4) (A5) (A6) (A7)};
\draw[black] plot [smooth cycle] coordinates {(C1) (C2) (C3) (C4) (C5) (C6) (C7)};
\draw[name path=normal line, red, ->, add=0 and 15] (B7) to (AA);
\fill[red,name intersections={of=curve 1 and normal line,total=\t}]
    \foreach \s in {1,...,\t}{(intersection-\s) circle (2pt)};
\draw node[anchor=south east] at (intersection-2) {$K$};
\draw node[anchor=south east] at (intersection-1) {$K-\zeta \nabla f(K)$};
\draw node at (BB) {$K-\eta \nabla f(K)$};
\end{tikzpicture}
\caption{Gradient descent interacting with the level curves of $f$ (\ref{f(k)}).}  \label{fig:level_curve}
  \end{figure}
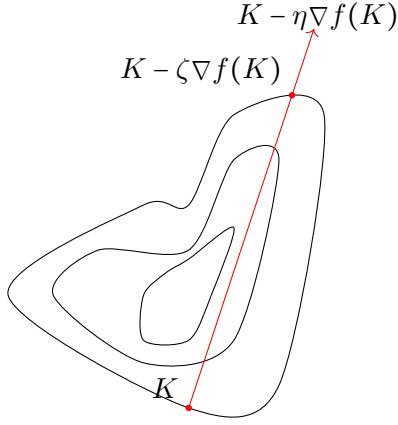
\begin{proof}
 Suppose that the sequence generated by the choice of $\eta_j$ is in fact stabilizing (to be proved subsequently!). This is crucial in our analysis as we use the Lyapunov matrix equation for the closed loop system, admitting a solution when $K_j$ is stabilizing; without this assumption, the matrix $X_j$ is not well-defined. By Lemma~\ref{lemma:value_difference}, we have,
  \begin{align*}
    f(K_{j+1}) - f(K_j) &= \Tr((X_{j+1}-X_j) {\bf \Sigma}) \\
    &= \Tr\left( Y_{j+1} \left( - 2\eta_j Y_j^{\top} {\bf N}_{j}^{\top} {\bf N}_j - 2\eta_j {\bf N}_j^{\top} {\bf N}_j Y_j  + 4 \eta_j^2 Y_j^{\top} {\bf N}_j^{\top} R {\bf N}_j Y_j  \right) \right) \\
                        &\le 4 \eta_j \Tr \left({\bf N}_j^{\top}{\bf N}_j (-Y_j Y_{j+1}  +\eta_j \lambda_n(R) Y_jY_{j+1}Y_j) \right),
    \end{align*}
    where the last equality follows from the cyclic property of matrix trace and Proposition~\ref{prop:linalg_facts}. 
In order to determine a stepsize $\eta_j$ such that $f(K_{j+1}) < f(K_j)$, we consider a univariate function,\footnote{We shall drop the indices in developing stepsizes for simiplicity.}
\begin{align*}
  g(\eta) = \Tr \left( {\bf N}^{\top} {\bf N} ( Y Y(\eta) - \eta \lambda_n(R) Y Y(\eta) Y)\right),
\end{align*} 
where $Y(\eta)$ is the solution of the matrix equation,\footnote{The function is not defined for every $\eta > 0$ but only for an interval for which $K-\eta 2MY$ is stabilizing.}
\begin{align*}
  Y(\eta) = (A-B(K-\eta 2{\bf N}Y )) Y(\eta) + Y(\eta) (A-B(K-\eta 2 {\bf N}Y))^{\top} + {\bf \Sigma}.
\end{align*}
Assuming that the choice of $\eta$ ensures staying in the sublevel set of $f(K)$, i.e., $f(K-\eta 2 {\bf N}Y) \le f(K)$,
we now examine whether $g(\eta) > 0$. By the Mean Value Theorem, we have
\begin{align*}
  g(\eta) = g(0) + \eta g'(\theta),
\end{align*}
for some $\theta \in [0, \eta]$; first note that,
\begin{align*}
  g'(\theta) = \Tr({\bf N}^{\top} {\bf N} ( Y Y'(\theta) - \lambda_n(R) Y Y(\theta) Y - \theta \lambda_n(R) Y Y'(\theta) Y)),
\end{align*}
and hence,
\begin{align*}
  g(\eta) &= \Tr( {\bf N}^{\top} {\bf N} (Y^2 + \eta Y Y'(\theta) - \eta \lambda_n(R) Y Y(\theta) Y - \eta \theta \lambda_n(R) Y Y'(\theta ) Y)) \\
  &= \Tr(Y {\bf N}^{\top} {\bf N} Y( I - \eta a Y(\theta) - \eta^2 \lambda_n(R) Y'(\theta))) + \eta \Tr({\bf N}^{\top}{\bf N}Y Y'(\theta) Y^{-1}Y) \\
  &\ge \Tr( Y {\bf N}^{\top} {\bf N} Y ) (1 - \eta \lambda_n(R) \lambda_n(Y(\theta)) - \eta^2 \lambda_n(R) \|Y'(\theta)\|2 - \eta  \|Y'(\theta)\|_2 \|Y^{-1}\|_2),
\end{align*}
where the last inequality follows from Von Neumann's trace inequality~\cite{horn2012matrix}.\footnote{An explicit form of the inequality we use here can be found in~\cite{mori1988comments}.} 
Noting that $\|Y^{-1}\|_2 = 1/\lambda_1(Y)$, ensuring that $g(\eta) > 0$ reduces to characterizing $\eta$ for which,
\begin{align*}
  1 - \eta a \lambda_n(Y(\theta)) - \eta^2 \lambda_n(R) \|Y'(\theta)\|_2 - \eta \frac{\|Y'(\theta)\|_2}{\lambda_1(Y)} > 0.
\end{align*}
The largest eigenvalue of $Y(\theta)$ and largest singular value of $Y'(\theta)$ over the sublevel set $\{K': f(K') \le f(K)\}$ can be bounded as,
\begin{align*}
  \lambda_n(Y(\theta)) \le \frac{f(K)}{\lambda_1(Q)}, \quad \|Y'(\theta)\|_2 \le \frac{4\|B{\bf N}Y\|_2f(K)}{\lambda_1(Q) \lambda_1({\bf \Sigma})};
\end{align*}
the proofs of the inequalities are deferred to Appendix~\ref{appendix:bound_Y_prime}.
Noting
\begin{align*}
  b \ge a \lambda_n(Y(\theta)) + \frac{\|Y'(\theta)\|}{\lambda_1(Y)},\qquad c  \ge a\|Y'(\theta)\|_2;
\end{align*}
it now suffices to determine $\eta$ such that $1 - b\eta - c \eta^2 > 0 $. As such, we require that,
\begin{align*}
  \eta < \sqrt{\frac{1}{c} + \frac{b^2}{4c^2}} - \frac{b}{2c}.
\end{align*}
It remains to show that if $\eta_j$ is chosen as above, our two opening assumptions are valid: (1) the sequence $\{K_j\}$ is stabilizing, and (2) $K_{j+1}$ remains in the sublevel set of $f(K_j)$. We prove these by contradiction. First, note that we can not have $K_{j+1}$ be stabilizing while $K_{j+1}\notin S_{f(K_j)}$. Suppose that this is the case. The sublevel set $S_{K_j} \coloneqq \{K: f(K) \le f(K_j)\}$ is compact and the ray $\{K_j - \zeta \nabla f(K_j) : \zeta \ge 0\}$ intersects the boundary of $S_{K_j}$ for some $\zeta > 0$; suppose that $K' = K_j - \zeta' \nabla f(K_j) \in \partial S_{K_j}$, where $\zeta'$ is the smallest positive real number for which this intersection occurs, i.e., the first time the ray intersects the boundary. 
It is clear $\zeta'$ must be greater than $\eta_j$ as otherwise we would have $\zeta' < \eta_j$ and $f(K_j - \zeta' \nabla f(K_j)) < f(K_j)$, a contradiction\footnote{Note what we proved above is: if a stepsize is strictly smaller than $\eta_j$, the function value is strictly decreasing if the gradient is not vanishing.}.  Now we prove that $K_j$ is stabilizing. If not, we must have $$[0, \eta_j) \subseteq [0, \zeta'],$$ since otherwise, there exists $s' < \eta_j$ such that $s' = \zeta'$ and $f(K') = f(K_j)$, which would also contradict the strict inequality $f(K') < f(K_j)$.
\end{proof}
With the lemma, it is straightforward to conclude the convergence rate.
\begin{theorem}
  \label{thrm:gd_linear}Putting $d_j = \max(b_j, c_j)$ where $b_j, c_j$ are given in Lemma~\ref{lemma:gd_function_decrease},
if $\eta_j = \sqrt{\frac{1}{3d_j} + \frac{1}{9}} - \frac{1}{3}$, we have,
\begin{align*}
  f(K_{j}) - f(K_*) \le q^j (f(K_0)-f(K_*)), \text{ and } \|K_j-K_*\|_F \le c_1 q^{j/2},
\end{align*}
where $q \in (0, 1)$ and $c_1 > 0$
are constants.
\end{theorem}
\begin{remark}
$\eta_j$ is acquired by noting that according to Lemma~\ref{lemma:gd_function_decrease},
\begin{align*}
  f(K_{j})-f(K_{j+1}) \ge 4\Tr(Y_j {\bf N}_j^{\top} {\bf N}_j Y_j)(\eta_j - d_j \eta_j^2 - d_j \eta_j^3).
  \end{align*}
  Maximizing $\eta_j-d_j \eta_j^2 - d_j \eta_j^3$ while esnuring $1-d_j \eta_j - d_j \eta_j^2 > 0$ yields the desired quantity. 
  \end{remark}
\begin{proof}
   Note the proposed stepsize rule satisfies $1-2 d_j \eta_j - 3 d_j \eta_j^2 =  0$. Putting $r_j = f(K_j)-f(K_*)$,
  we observe that with the chosen stepsize $\eta_j$,
\begin{align*}
  r_{j} - r_{j+1}  
                  &\ge 4\Tr(Y_j M_j^{\top}M_j Y_j) (d_j \eta_j^2 + 2d_j^2 \eta_j^3 ) \eqqcolon \nu_j r_j.
\end{align*}
It follows that,
\begin{align*}
  r_{j+1} \le (1-\nu_j) r_j \eqqcolon q_j r_j.
\end{align*}By Proposition~\ref{prop:stepsize_bound}, the proposed stepsize in Lemma~\ref{lemma:gd_function_decrease} is bounded below, i.e., not vanishing.
Hence, the sequence $\{q_j\}$ is upper bounded away from $1$\footnote{It is rather clear $d_j$ is lower bounded away from $0$. So $d_j \eta_j^2 + 2 d_j^2 \eta_j^3 > 0$.},
\begin{align*}
  q_j \le q < 1.
  \end{align*}
Thereby, 
\begin{align*}
  f(K_{j}) - f(K_*) \le q^j \left(f(K_0)-f(K_*)\right).
\end{align*}
To show the convergence of the iterates, we first observe that,
\begin{align*}
  \|K_{j+1} - K_j\|_F^2 &= \eta_j^2 \|\nabla f(K_j)\|_F^2 \le \frac{\eta_j^2}{\tau} r_j \\
                      &\le \frac{\eta_j}{\tau} q^j r_0,
\end{align*}
with $\tau$ is as in (\ref{tau}).
It is clear the sequence $\{\eta_j\} \subseteq \bb R_+$ is upper bounded, denoting as $\mu$, namely $\mu \ge \eta_j$ for every $j$.
The sequence of iterates $\{K_j\}$ is thus Cauchy and converges to some stationary point; however, there is only one stationary point $K_*$. This implies that $\lim_{j \to \infty} K_j = K_*$ and hence,
\begin{align*}
  \|K_j - K_*\|_F &= \lim_{n \to \infty} \|K_j - K_n\|_F \le \sum_{j=n}^{\infty}\|K_{j+1}-K_{j}\|_F \\
  &\le  \sqrt{\frac{\mu}{\tau}} r_0 \sum_{j=n}^{\infty} q^{j/2} = \frac{\sqrt{\frac{\mu}{\tau}} r_0 }{1-\sqrt{q}} \, q^{j/2}.
\end{align*}
\end{proof}
    \begin{remark}
In our simulations, the linear rate is much better than what is estimated by the above result. 
         \end{remark}
\section{Natural Gradient Flow on $\ca H$}
\label{sec:ngf}
If we inspect the proof of the Lyapunov stability of the gradient system (Theorem~\ref{thrm:exponential_trajectory}), the positive definite matrix $Y$ does not affect the qualitative nature of these properties. Nevertheless, the matrix $Y$ introduces a constant factor in the corresponding upper bounds. In this section, we consider a family of gradient systems of the form,
\begin{align}
  \label{eq:ngf} \dot{K}_t &= -\nabla f(K_t) Y_t^{-\gamma} = -2{\bf N}_tY_t^{1-\gamma}, 
\end{align}
where $\gamma > 0$ is (real) scalar, ${\bf N}_t$ and $Y_t$ are simplified notation for ${\bf N}_{K_t}$ and $Y_{K_t}$.
As discussed subsequently,  such parameterized gradient system can achieve better convergence rate for different values of $\gamma$.
Viewing such a gradient flow in the context of a flow on a Riemannian manifold is particularly pertinent.\footnote{We will see that in our case, it is better to choose $\gamma$ other than $\gamma=1$.}
In fact, as $\ca H$ is open, it is a \emph{submanifold} in $\bb M_{m \times n}(\bb R)$. The inner product induced by $Y^{\gamma}$, i.e., $\langle M, N \rangle_{Y^{\gamma}} = \Tr(M^{\top} NY^{\gamma})$ is a well-defined Riemannian metric over $\ca H$. 
\begin{proposition}
  Over $\ca H$, the inner product $\langle \cdot, \cdot\rangle_{Y(K)^\gamma}$ induces a Riemannian metric.
\end{proposition}
\begin{proof}
  Note that $Y(K)$ is positive definite for every $K \in \ca H$.
  It suffices to show that $Y(K)$ varies smoothly with $K$. But this follows from,
  \begin{align*}
    \vect(Y) = (I \otimes A_K - A_K \otimes I)^{-1} \vect({\bf \Sigma}).
  \end{align*}
\end{proof}
We can thus view $\ca H$ as a Riemannian manifold with metric induced by $\langle \cdot, \cdot \rangle_{Y^{\gamma}}$; the function $f: \ca H \to \bb R$ is then a scalar-valued function defined on this manifold. Let us now consider the gradient of $f$, denoted by $\text{grad} f$, with respect to the Riemannian metric induced by $\langle \cdot, \cdot\rangle_{Y^{\gamma}}$ on $\ca H$\footnote{We will use standard notations in Riemmaninan manifold theory~\cite{tu2017differential}. For example, $df$ will denote $1$-form and $\text{grad} f$ will denote the gradient with respect to a Riemannian metric. As we are working in Euclidean space, we implicitly identiy all tangent vectors by stardard isomorphism, i.e., $T_K {\ca H} \approx \bb M_{m \times n}(\bb R)$.}.
\begin{proposition}
  Over the Riemannian manifold $\left(\ca H, \langle \cdot, \cdot\rangle_{Y^{\gamma}}\right)$, $\text{grad} f = 2(RK-B^{\top} X) Y^{1-\gamma}$.
\end{proposition}
\begin{proof}
  It suffices to note that,
  \begin{align*}
    df(K)[E] = 2\Tr(E^T {\bf N}_K Y) = \langle E, 2{\bf N}_KY^{1-\gamma}\rangle_{Y^{\gamma}}.
  \end{align*}
\end{proof}
Now the gradient flow of interest on this manifold is,
\begin{align*}
  \dot{K}_t = - \text{grad} f(K_t).
\end{align*}
We observe that with respect to the Riemannian metric, the potential function decays at an exponential rate (compare the difference with the gradient flow in Lemma~\ref{lemma:convergence_energy}).
\begin{lemma}
  \label{lemma:convergence_energy_ngf}
  For $K_0 \in \ca H$, denote $K(t)$ as the solution of~\eqref{eq:ngf}. Then
\begin{align*}
  f(K_t) - f(K_*) \le e^{- r t} (f(K_0)-f(K_*)),
\end{align*}
where $r$ is a constant determined by the system parameters $A,B,Q,R$ and $K_0$.
\end{lemma}
\begin{proof}
  The proof proceeds similar to Lemma~\ref{lemma:convergence_energy}. We only need to note that with respect to the Riemannian metric,
  \begin{align*}
    \dot{V}(K_t) &= df(K_t) \dot{K}(t) = \langle \text{grad} f(K_t), \dot{K_t}\rangle_{Y^\gamma} = - 4\Tr({\bf N}_t^{\top} {\bf N}_tY_t^{2-\gamma}).
  \end{align*}
 According to Lemma~\ref{lemma:natural_gradient_dominated}, we now have,
  \begin{align*}
    \dot{V}(K_t) \le -\frac{\lambda_1(R)\lambda_1(Y_t^{2-\gamma})}{\lambda_n(Y_*) } V(K_t) \le -r V(K_t),
  \end{align*}
  where
  \begin{align*}
    r = \min_{K \in S_{f(K_0)}} \frac{\lambda_1(R) \lambda_1(Y_t^{2-\gamma})}{\lambda_n(Y_*)}.
    \end{align*}
    Note the compactness of the sublevel set $S_{f(K_0)}$ justifies the $\min$ operation and guarantees $r>0$.
  \end{proof}
Over the Riemannian manifold, the Lyapunov functional converges exponentially to the origin via the natural gradient flow, which leads to an exponentially stable trajectory.
\begin{theorem}
  Over $\left(\ca H, \langle \cdot, \cdot\rangle_{Y^{\gamma}}\right)$, for the natural gradient flow~\eqref{eq:ngf}, the energy functional $f(K_t)-f(K_*)$ converges exponentially to the origin. Moreover, the trajectory $K_t$ is exponentially stable in the sense of Lyapunov.
\end{theorem}
\begin{proof}
Over the Riemannian manifold $\left(\ca H, \langle \cdot, \cdot\rangle_{Y^{\gamma}}\right)$, we have,
\begin{align*}
  f(K_t) - f(K_*) &= \int_{t}^{\infty} \langle \text{grad} f, \text{grad} f \rangle_{Y^{\gamma}} dt \\
                  &=\int_t^{\infty} 4\Tr(M_t^{\top} M_t Y_t^{2-\gamma})dt.
\end{align*}
Now by a smiliar computation in Theorem~\ref{thrm:gf_exponential_trajectory}, we have
\begin{align*}
  \Tr(Y_t^{1-\gamma}{\bf N}_t^{\top}(K_t-K_*)) \le c \Tr({\bf N}_t^{\top}{\bf N}_t Y_t^{2-\gamma}),
  \end{align*}
  where $c > 0$ is some constant. So
  \begin{align*}
  f(K_t) - f(K_*) &\ge \frac{4}{c} \int_{t}^{\infty} \langle {\bf N}_tY_t^{1-\gamma}, K_t-K_* \rangle_{Y^{\gamma}} dt \\
                  &\ge -\frac{2}{c} \|K_t-K_*\|_F^2 \big\vert_{t}^{\infty} \\
                  &= \frac{2}{c}\|K_t-K_*\|_F^2.
    \end{align*}
Hence, $\|K_t -K_*\|^2 \le V(K_t) \le e^{-r t} V(0)$ with $r = \frac{2}{c}$.
\end{proof}
\begin{remark} \label{remark:gf_vs_ngf}
  We note that the convergence rate of trajectory $K_t$ is dependent on $\lambda_1(Y)$ and $\lambda_n(Y)$. For example, when $\gamma=1$ and ${\bf \Sigma} = 2I$, then the natural gradient flow converges faster than the gradient flow since $\lambda_1(Y) > 1$. On the other hand, if $\gamma = 1$ and $\lambda_1(Y) < 1$, then gradient flow converges faster than natural gradient flow.\footnote{This can be done by an ${\bf \Sigma}$ that has a spectrum bounded by $1$.}
   Simulation results in~\S\ref{sec:numerical_results}
   show that this parameterized gradient flow offers a significant computational advantage for LQR.
\end{remark}
{We remark that in the particular case of $\gamma = 1$, the natural gradient flow has a favorable property with respect to the induced flow on the value matrix $X_t$. Consider again the flow, 
  \begin{align}
    \label{eq:ngf_1}
    \dot{K}_t = - 2(RK_t - B^{\top} X),
    \end{align}
inducing the flow over the ``value'' matrix $X_t \coloneqq X(K_t)$ given by,
    \begin{align}
      \label{eq:ngf_cone}
      \dot{X}_t = \frac{d X_t}{d K} \dot{K}_t.
      \end{align}
      \begin{lemma}
        \label{lemma:mono_X_t}
        For $K_0 \in \ca H$, the gradient flow~\eqref{eq:ngf_1} induces a well-posed flow over the positive semidefinite cone $X_t$~\eqref{eq:ngf_cone}. Moreover, the trajectory $\{X_t\}$ is monotonically decreasing in Loewner ordering.
        \end{lemma}
        \begin{proof}
          The well-posedness follows from the well-posedness of $\{K_t\}$. To show that the trajectory is monotonically decreasing, it suffices to observe,
          \begin{align*}
            \dot{X}_t = \frac{d X_t}{d K} \dot{K}_t 
                      &= \int_{0}^{\infty} e^{A_{K_t}^{\top}t}  \left( \dot{K}_t^{\top} M_t + M_t^{\top} \dot{K}_t \right) e^{A_{K_j} t}dt\\
                      &=- \int_{0}^{\infty} e^{A_{K_t}^{\top}t} 4M_t^{\top} M_t) e^{A_{K_j} t}dt
 \preceq 0,
            \end{align*}
            where the second inequality follows from~\eqref{eq:X_derivative}.
          \end{proof}
      Note that this monotonicity does not hold in general for gradient flow: in this case the flow is dictated by ${\bf \Sigma}$ and along the trajectory, one can only guarantee that the function value $\Tr(X_t {\bf \Sigma})$ decreases.
}
{
    \subsection{Discretization of Natural Gradient Flow}
    \label{sec:ngd}
    In this section, we delve into the discretization of natural gradient flow; we shall only consider the case when $\gamma = 1$.\footnote{Other choices can be analyzed in a similar manner.} Specifically, we consider the gradient flow,
\begin{align*}
  \dot{K}_t = -2(RK_t - B^{\top} X).
\end{align*}
The forward Euler discretization yields,
\begin{align}
  \label{eq:ngd}
  K_{j+1} = K_j - 2 \eta_j{\bf N}_j,
\end{align}
where $\eta_j$ is the stepsize to be determined and recall ${\bf N}_j = RK_j - B^{\top}X_j$.
In discretizing gradient flow, our guideline is to choose a stepsize such that the function value is sufficiently decreased while keeping iterates stabilizing. However, in natural gradient flow with $\gamma = 1$, we observe that by Lemma~\ref{lemma:mono_X_t}: if we follow the natural gradient flow, the value matrix is monotonic with respect to the semidefinite cone. This essentially means that taking a sufficiently small stepsize in the direction of the natural gradient would guarantee a decrease in the value of the Lyapunov matrix solution $X_{t + \delta} \preceq X_{\delta}$. 
\begin{lemma}
  \label{lemma:ngd_value_matrix_difference}
Consider the sequence $\{K_j\}$ generated by~\eqref{eq:ngd}. Denote by $\{X_j\}$ the corresponding Lyapunov matrix solution with respect to $K_j$.
  If $\eta_j < {1}/{\lambda_n(R)}$, then $K_j$ is stabilizing for every $j \ge 0$ and $X_{j+1} \prec X_j$. In particular,
  $Z \coloneqq X_{j+1}-X_j$ solves the Lyapunov matrix equation,
\begin{align*}
  A_{K_{j+1}}^{\top}Z + Z A_{K_{j+1}}^{\top} +  {\bf N}_{j}^{\top}(-4\eta_j I + 4\eta_j^2 R ) {\bf N}_{j}  = 0.
\end{align*}
\end{lemma}
\begin{proof}
  The proof proceeds similarly to Lemma~\ref{lemma:gd_function_decrease}.
  First, we suppose that the sequence generated by the choice of $\eta_j$ is in fact stabilizing (to be proved subsequently). By Lemma~\ref{lemma:value_difference},
      \begin{equation} \label{eq:value_difference}
      \begin{split}
        &{A}_{K_{j+1}}^{\top} (X_{j+1}-X_j)  - (X_{j+1}-X_{j}) {A}_{K_{j+1}} + (K_{j+1}-K_{j})^{\top}{\bf N}_j + {\bf N}_j^{\top}(K_{j+1}-K_{j}) \\
&+ (K_{j+1}-K_{j})^{\top} R (K_{j+1}-K_j) = 0.
      \end{split}
   \end{equation}
   If $K_{j+1} - K_j = - 2\eta_j {\bf N}_j$, then
\begin{align*}
        {A}_{K_{j+1}}^{\top} (X_{j+1}-X_j)  - (X_{j+1}-X_{j}) A_{K_{j+1}}+ {\bf N}_j^{\top}(-4\eta_j I + 4\eta_j^2 R ) {\bf N}_j = 0. 
\end{align*}
Hence, if $-4\eta_j I + 4\eta_j^2 R  \prec 0$, then $X_{j+1} \prec X_j$. This can be guaranteed by choosing:
\begin{align*}
  \eta_j < \frac{1}{\lambda_n(R)}.
\end{align*}
It now remains to show that if $\eta_j$ is chosen as above, the sequence will be stabilizing. Suppose that $K_j$ is stabilizing. Note that the sublevel set $\ca S_{K_j} \coloneqq \{K: f(K) \le f(K_j)\}$ is compact and the ray $K_j - \zeta M_j$ intersects the boundary of $\ca S_{K_j}$ for some $\zeta=\zeta' >0$; suppose that $K' = K_j - \zeta' M_j \in \partial \ca S_{K_j}$. But this implies that $$[0, \frac{1}{\lambda_n(R)}) \subseteq [0, \zeta'],$$ since otherwise, there would exist $s' < {1}/{\lambda_n(R)}$ such that $s' = \zeta'$ and $f(K_j - s' M_j) = f(K_j)$, contradicting $f(K_j - s' M) < f(K_j)$.
\end{proof}
The problem of determining the optimal stepsize can be done by minimizing the expression,
\begin{align*}
  -4\eta_j I + 4\eta_j^2 (R) \preceq 0,
\end{align*}
over the positive semidefinite cone. This is equivalent to minimizing,
\begin{align*}
  -4\eta_j + 4\eta_j^2 \left(\lambda_n(R)\right),
\end{align*}
at $\eta_j \in [0, 1/\lambda_n(R))$. Obviously, the optimal stepsize should be $\eta_j ={1}/{(2\lambda_n(R))}$. With this choice of stepsize, the function value converges linearly to the optimal value function.
\begin{theorem}
  \label{thrm:ngd_convergence}
  If $\eta_j = {1}/({2\lambda_n(R)})$, we have,
\begin{align*}
  f(K_{j}) - f(K_*) \le q_0^j (f(K_0)-f(K_*)), \text{ and } \|K_{j} - K_*\|_F \le c_2 q_0^{j/2}.
\end{align*}
where
$q_0 = 1- {4 \mu\lambda_1(R)})/({\lambda_n(Y_*) \lambda_n(R)}$ with $\mu  = \min_{K \in S_{f(K_0)}} \lambda_1(Y(K))$ and $c_2$ is some positive constant.
\end{theorem}
\begin{proof}
  First note by Lemma~\ref{lemma:ngd_value_difference}, the sequence of value matrices $\{X_j\}$ is monotonically decreasing and bounded below. Thus $X_j \to X_*$ as $j \to \infty$. We now characterize the convergence rate.
  Putting $r_j = f(K_j)-f(K_*)$, we observe that with the chosen $\eta_j$,
\begin{align*}
  r_{j} - r_{j+1} &= \Tr((X_{j} - X_{j+1}) {\bf \Sigma}) \ge \Tr(\frac{1}{\lambda_n(R)} {\bf N}_j^{\top} {\bf N}_j Y_{j+1}) \ge \frac{\lambda_1(Y_{j+1})}{\lambda_n(R)} \Tr(M_j^{\top} M_j) \ge \frac{4 \lambda_1(Y_{j+1})\lambda_1(R)}{\lambda_n(Y_*) \lambda_n(R)} r_j.
\end{align*}
It thus follows that,
\begin{align*}
  r_{j+1} \le \left(1-\frac{4 \lambda_1(Y_{j+1}) \lambda_1(R)}{\lambda_n(Y_*) \lambda_n(R)}\right) r_j \eqqcolon q_j r_j.
\end{align*}
Putting $\mu = \min_{K \in S_{f(K_0)}} \lambda_1(Y(K))$\footnote{This is justified by the fact that $X_j$ is monotonically decreasing and thus stays in the compact sublevel set.}, we have
\begin{align*}
  q_j &= 1- \frac{4 \lambda_1(R) \lambda_1(Y_{j+1})}{\lambda_n(Y_*) \lambda_n(R)} \le 1 - \frac{4 \lambda_1(R) \mu}{\lambda_n(Y_*) \lambda_n(R)}.
\end{align*}
Thereby, 
\begin{align*}
  f(K_{j}) - f(K_*) \le q_0^j \left(f(K_0)-f(K_*)\right).
\end{align*}
The proof to the convergence of the iterates is almost identical to the one in Theorem~\ref{thrm:gd_linear}.
\end{proof}
}
{
  \section{Quasi-Newton Flow on $\ca H$}
  \label{sec:qnf}
  In this section, we motivate a quasi-Newton flow over the set of stabilizing feedback gains (policy) $\ca H$.\footnote{The justification for calling this evolution a quasi-Newton flow becomes apparent subseqeuntly.} As observed previously, the Hessian of the LQR cost $f(K)$ is not positive definite everywhere. As such, there is no well-defined notion of (global) Newton iteration over policy space. However, examining Lemmas~\ref{lemma:value_difference} and \ref{lemma:ngd_value_matrix_difference} allows us to derive a local second-order approximation of the LQR cost under the Riemannian metric $Y$. With this metric, recall that the gradient of $f$ is,
  \begin{align*}
    \text{grad} f(K) = 2(RK-B^{\top} X).
    \end{align*}
    We now provide the second-order approximation of the cost function.
  \begin{lemma}
    \label{lemma:second_order}
  When $K$ and $K+\Delta K$ are both stabilizing for sufficiently small $\Delta K$,\footnote{By openness of $\ca H$, if $\Delta K$ is sufficiently small, $K+\Delta K$ is stabilizing provided that $K$ is.} then,
    \begin{align*}
      f(K+\Delta K) = f(K) + \langle \text{grad} f(K), \Delta K \rangle_Y + \langle \Delta K, R (\Delta K)\rangle_Y + \ca R(\Delta K),
    \end{align*}
    where $\|\ca R(\Delta K)\|$, the remainder of the approximation, is $O(\|\Delta K\|^2)$.
  \end{lemma}
  \begin{proof}
    Suppose that $X_{K+\Delta K}$ and $X_K$ are the corresponding value matrices for $K+\Delta K$ and $K$, respectively.
    By Lemma~\ref{lemma:ngd_value_matrix_difference}, we have,
      \begin{align}
        \label{eq:second_order} 
        \begin{split}
        f(K+\Delta K) - f(K) &= \Tr( Y_{K+ \Delta K} (2M_K \Delta K + (\Delta K)^{\top} R \Delta K)).
                             \end{split}
        \end{align}
        where $Y_{K+\Delta K}$ solves the Lyapunov equation,
        $$A_{K+\Delta K} Y_{K+\Delta K} A_{K + \Delta K}^{\top} + {\bf \Sigma} - Y_{K + \Delta K} = 0;$$
        By continuity of $Y$ with respect to $K$, for any matrix norm $\|\cdot \|$, we recongonize
        \begin{align*}
          \|Y_{K +\Delta K} - Y_K\| \approx O(\|\Delta K\|).
          \end{align*}
          It follows
          \begin{align*}
            f(K+\Delta) \approx f(K) + \langle \Delta K, \text{grad} f(K)\rangle_{Y} + \langle \Delta K, R (\Delta K)\rangle_Y.
            \end{align*}
    \end{proof}
    Lemma~\ref{lemma:second_order} essentially states that we have a somewhat ``good'' local second-order approximation of $f(K)$ with respect to the Riemannian metric $Y$. We may now devise a flow to minimize $f(K)$ by minimizing this second-order approximation, namely,
  \begin{align*}
    \dot{K}_t = R^{-1} \text{grad} f(K_t) = R^{-1}(RK_t-B^{\top} X_t) = K_t - R^{-1} B^{\top}X_t. 
  \end{align*}
 The analysis presented in \S\ref{sec:gf} and \S\ref{sec:ngf} allow us to obtain a streamlined proof of the convergence of this flow; as such, we omit the proof.
\subsection{Discretization of Quasi-Newton Flow}
\label{sec:qn_iteration}
The quasi-Newton flow over $\ca H$ has interesting consequences in terms of its discretization: the forward Euler leads to the iterative procedure 
\begin{align}
  \label{eq:qnd}
K_{j+1} = K_j - \eta_jR^{-1} \text{grad} f(K_j) = K_j - \eta_j R^{-1} (2(RK_j-B^{\top} X_j))
\end{align}
 with stepsize $\eta_j$ to be determined; we shall show that with constant stepsize $\eta = \frac{1}{2}$, both the function value and the iterates will converge quadratically to the optima. 
We first observe that if $\eta < 1$, the corresponding sequence of value matrices $\{X_j\}$ is monotonically decreasing over the positive semidefinite cone.
        \begin{lemma}
  \label{lemma:newton_value_matrix_difference}
Consider the sequence $\{K_j\}$ generated by~\eqref{eq:qnd}. Denote by $\{X_j\}$ the corresponding Lyapunov matrix solution with respect to $K_j$.
  If $\eta_j < 1$, then $K_j$ is stabilizing for every $j \ge 0$ and $X_{j+1} \preceq X_j$. In particular
  $Z \coloneqq X_{j+1}-X_j \preceq 0$ solves the Lyapunov matrix equation,
\begin{align*}
  A_{K_{j+1}}^{\top} Z + Z A_{K_{j+1}}  +  (-4 \eta_j + 4 \eta_j^2){\bf N}_{j}^{\top} R^{-1} {\bf N}_{{j}}  = 0.
\end{align*}
\end{lemma}
\begin{proof}
  Suppose that with $\eta_j < 1$, the sequence generated by~\eqref{eq:qnd} are all stabilizing.\footnote{Similar to the proof to Lemma~\ref{lemma:ngd_value_matrix_difference}, we need this assumption to make sense of defining the corresponding value matrix sequence $\{X_j\}$.}
  Substituting the update rule~\eqref{eq:qnd} in~\eqref{eq:value_difference} yields,
  \begin{align*}
    A_{K_{j+1}}^{\top} (X_{j+1} - X_j) + (X_{j+1}-X_j)A_{K_{j+1}} + (-4 \eta_j + 4\eta_j^2 ){\bf N}_j^{\top} R^{-1} {\bf N}_j = 0.
  \end{align*}
  It is now clear if $\eta_j < 1$, then $X_{j+1}-X_j \preceq 0$. To show the choice of $\eta_j$ guaranteeing the stability of $A-BK_j$, we may follow almost the same argument as in the proofs of Lemmas~\ref{lemma:value_difference} and~\ref{lemma:ngd_value_matrix_difference}. 
    \end{proof}
The optimal stepsize for the quasi-Newton iteration is obtained by minimizing the quantity $-4 \eta + 4 \eta^2$. As such, the optimal stepsize is $\eta_j = 1/2$ for every $j$. The corresponding update is then equivalent to,
  \begin{align}
    \label{eq:quasi-newton}
    K_{j+1} &= K_j - \frac{1}{2} R^{-1} 2(RK_j - B^{\top} X_j) = K_j - K_j + R^{-1} B^{\top} X_j = R^{-1} B^{\top} X_j. 
    \end{align}
    \begin{remark}
      With the optimal choice of stepsize as $\eta = 1/2$, the quasi-Newton over $K$ coincides with the Kleinman-Newton algorithm~\cite{kleinman1968iterative}, obtained by considering the Newton iteration over the ARE. We have thus provided an alternative point view of this algorithm: the algorithm can be obtained directly over the policy space even without the ARE.
      \end{remark}
    \begin{theorem}
      \label{thrm:qn_convergence}
      With stepsize $\eta = 1/2$, the update~\eqref{eq:quasi-newton} converges to the global minimum at a Q-quadratic rate. Namely, there exists constants $c > 0, c_3 > 0$, such that,
      \begin{align*}
        f(K_j) - f(K_*) \le c (f(K_{j-1}) - f(K_*))^2  \quad \text{ and } \quad \|K_j - K_*\|_F \le c_3 \|K_{j-1} - K_* \|_F^2.
      \end{align*}
    \end{theorem}
    \begin{proof}
      By Lemma~\ref{lemma:newton_value_matrix_difference} and noting $RK_* - B^{\top} X_* = 0$, we have
      \begin{align}
        \label{eq:qn_myeq1}
        X_{j+1} - X_* = \int_{0}^{\infty} e^{A_{j+1}^{\top}t} (K_{j+1}-K_*)^{\top} R (K_{j+1} - K_*) e^{A_{j+1}t} t.
      \end{align}
       It then follows that,
      \begin{align*}
        f(K_{j+1}) - f(K_*) &= \Tr( (X_{j+1} - X_*) {\bf \Sigma}) \\
                            &=\Tr(Y_{j+1} (K_{j+1}-K_*)^{\top} (K_{j+1}-K_*))\\
                            &\le \|Y_{j+1}\|_2 \|R\|_2 \Tr((K_{j+1} - K_*)^{\top} (K_{j+1} - K_*)).
      \end{align*}
      But we note
      \begin{align*}
        K_{j+1}-K_* = R^{-1} B^{\top}(X_{j}-X_*).
       \end{align*}
      Hence,
       \begin{align*}
         \|K_{j+1} - K_*\|_F \le \|R^{-1}\|_F\|B\|_F \|X_{j}-X_*\|_F.
        \end{align*}
          Consequently,
          \begin{align*}
            f(K_{j+1}) - f(K_*) &\le \|R\|_2\|Y_{j+1}\|_2 \|K_j - K_*\|_F^2 \\
                                &\le \|R\|_2 \|Y_{j+1}\|_2 \|B\|_F\|R^{-1}\|_F \|X_j-X_*\|_F^2 \\
                                &\le \|R\|_2 \|B\|_F\|Y_{j+1}\|_2 \|R^{-1}\|_F \frac{1}{\lambda_1^2({\bf \Sigma})} \left( \Tr((X_j-X_*){\bf \Sigma})\right)^2 \\
            &\eqqcolon c \left(f(K_j)-f(K_*)\right)^2.
            \end{align*}
            Note the sequence of value matrix $\{X_j\}$ is monotonically decreasing in the PSD cone. Thus $\|Y_{j+1}\|_2$ can be bounded by
\begin{align*}
  \|Y_{j+1}\|_2 \le \Tr(Y_{j+1}) \le \frac{f(K_{j+1})}{\lambda_1(Q)} \le \frac{f(K_0)}{\lambda_1(Q)}.
  \end{align*}
            To establish the quadratic convergence of iterates, we observe by Proposition~\ref{prop:linalg_facts} and equation~\eqref{eq:qn_myeq1}
\begin{align*}
  \Tr(X_{j+1} - X_*) &\ge \lambda_1(R) \lambda_1(Y_*) \|K_{j+1} - K_*\|_F^2, \\
  \Tr(X_{j+1} - X_*) &\le \|R\|_2 \|Y_*\|_2 \|K_{j+1} - K_*\|_F^2.
\end{align*}
On the other hand, 
\begin{align*}
  \Tr((X_{j+1} - X_*) {\bf \Sigma}) &\le c(\Tr((X_j-X_*){\bf \Sigma}))^2 \\
&\le c (\|{\bf \Sigma}\|\Tr(X_j-X_*) )^2 \\
&\le c(\|{\bf \Sigma}\|  \|R\|_2 \|Y_*\|_2 \|K_{j} - K_*\|_F^2)^2.
\end{align*}
It follows,
\begin{align*}
  \|K_{j+1}-K_*\|_F^2 \le \frac{c \|{\bf \Sigma}\|^2 \|R\|_2^2 \|Y_*\|_2^2}{\lambda_1({\bf \Sigma}) \lambda_1(R) \lambda_1(Y_*)} \|K_j-K_*\|_F^4.
\end{align*}
    \end{proof}
\section{Structured Controller Sythesis}
In this section, we concern the problem of designing feedback gain $K$ with linear structures. In particular, we are mostly interested in those feedback gains with an arbitrary zero pattern. This is a natural formulation of distributed networked systems modeled on a communication graph $\ca G=(V, E)$. In such setting, structured feedback gains reflecting the underlying interaction network are of particular interest. If a subset of agents are accessible to be controlled upon and the control law must only utilize the information of an agent and its neighbors, the feedback gains must have a zero pattern that is compatible with communication graph, i.e., $K_{ij} = 0$ if $(i,j) \notin E(\ca G)$. We shall \emph{emphasize} in the scenario\footnote{Indeed, the most widely used model in networked systems.} that the interaction network is modeled by a graph $\ca G=(V, E)$, each agent can only have direct control over its own dynamics, using information from their own sensors and from communicating with neighboring agents, i.e., $B$ has a diagonal structure.
If all agents have their own control over their dynamic, without loss of generality, we suppose that $B=I$. If only subsets of agents have direct control over their dynamic, without loss of generality, by permuting the agents, we may assume
\begin{align*}
B=
  \begin{pmatrix}
    I_{m \times m} \\
    \bf 0_{(n-m) \times m}
  \end{pmatrix}.
\end{align*}
So we are interested in optimizing the cost function $f(K)$ over the set
\begin{align*}
  \ca K= \{K \in \ca U: A-BK \in \ca H\},
\end{align*}
where $\ca U$ is a linear subspace defined by the graph structure, i.e.,
\begin{align*}
  \ca U =\{M \in \bb M_{n \times m}(\bb R): M_{i,j} = 0 \text{ if and only if }(i,j) \not\in E(\ca G)\}.
\end{align*}
Projected gradient descent (PGD) is a natural choice in acquiring feedback gains in the set $\ca K$. 
It refers to the iterative procedure
\begin{align}
  \label{eq:pgd}
  K_{j+1} = P_{\ca K}(K_j - t \nabla f(K_j)),
\end{align}
where $t$ is stepsize and will be determined in Lemma~\ref{lemma:sublinear_pgd}. One may immediately notice that the geometry of $\ca K$ could be rather complicate. Indeed, this set could have exponentially many path connected components (see~\cite{bu2019topological_mimo}). But some favorable structures of $A$ and the graph $\ca G$ would guarantee $\ca K$ has only $1$ connected component as pointed out in~\cite{bu2019topological_mimo}. This is out of the scope of this manuscript. We shall mainly concern how to update $K$ in the path connected component we initialize.\par
But even this modest goal faces difficulty: $\ca K$ has complicate geometry and one will ask how to efficiently project onto $\ca K$. We shall show next the seemingly relaxed updating rule
\begin{align*}
  K_{j+1} = P_{\ca U}(K_j - t \nabla f(K_j))
\end{align*}
is equivalent to \eqref{eq:pgd} where $P_{\ca U}$ is an orthogonal projection onto the space $\ca U$. 
\begin{theorem}
  \label{thm:pgd}
  The updating rule \eqref{eq:pgd} is equivalent to
\begin{align*}
  K_{j+1} = K_j - t P_{\ca U}(\nabla f(K_j))
\end{align*}
provided the initial condition $K_0 \in \ca K$.
\end{theorem}
The proof is almost verbatim to the proof to Theorem VI.1 in~\cite{bu2019lqr_calculus}. If we concern the restriction $g \coloneqq f|_{\ca K}$, it is also shown in~\cite{bu2019lqr_calculus} the procedure~\eqref{eq:pgd} is equivalent to the gradient descent on $g$. $g$ is coercive in its own right with $\nabla g(K) = P_{\ca U}(\nabla f(K))$ and $\nabla^2 g(K) = \nabla^2 f(K)$ for $K \in \ca K$. It is now clear by picking a constant step size $1/L$ with $L = \sup_{K \in S_{g(K_0)}} \|\nabla^2 g(K)\|$, the procedure converges to the first-order stationary point at a sublinear rate.
\begin{lemma}
  \label{lemma:sublinear_pgd}
  Suppose $K_0 \in \ca K$ and recall the sublevel set is given by $S_{g(K_0)}=\{K \in \ca K: g(K) \le g(K_0)\}$.
  Putting $L = \sup_{K \in S_{g(K_0)}} \|\nabla^2 g(K)\|$, if the stepsize $t$ in~\eqref{eq:pgd} is set to be $t = 1/L$, the sequence $\{K_j\}_{j=0}^{\infty}$ generated by projected gradient descent \eqref{eq:pgd} convergences to a first-order stationary point at a sublinear rate, i.e.,
  \begin{align*}
    \|\nabla g(K_j) \|^2 \to 0 
  \end{align*}
  at a sublinear rate of $O(1/k)$.
\end{lemma}
\begin{proof}
  This is straightforward by Lemma~\ref{lemma:sequence_bounded}, Theorem~\ref{lemma:global_convergence} and Section $1.2.3$ in~\cite{nesterov2013introductory}.
\end{proof}
We now observe $L$ can be upper bounded (corresponding to a lower bound of step size) by system parameters $A, B$, LQR weighting matrices $Q, R$, and initial conditon $K_0$.
\begin{lemma}
  On sublevel set $S_{g(K_0)}$, we have
  \begin{align*}
    \sup_{K \in S_{f(K_0)}}\|\nabla^2 g(K)\| \le \sup_{K \in S_{g(K_0)}}\| \nabla^2 f(K)\|.
  \end{align*}
\end{lemma}
\begin{proof}
  We only need to observe for each $K \in \ca K$,
  \begin{align*}
    \|\nabla^2 g(K)\| &= \sup_{\|E\|_F=1, E \in \ca U} \langle \nabla^2 g(K)[E], E\rangle \\
                      &\le \sup_{\|E\|_F=1} \langle \nabla^2 g(K)[E], E\rangle \\
                      &= \sup_{\|E\|_F=1} \langle \nabla^2 f(K)[E], E\rangle.
  \end{align*}
\end{proof}Conceptually, starting from $K_0$, as the sublevel set $S_{f(K_0)}$ is compact and the operator norm of the Hessian is continuous, it achieves maximums $L = \sup_{K \in S_{f(K_0)}} \|\nabla^2 f(K)\|$; this in turn implies the gradient mapping $\nabla f(K)$ is Lipschitz. As a common practice in convex optimization, we might choose constant stepsize $\eta = 1/L$. One must take more care in this scenario, since the sublevel set is not convex, it is not clear that $K_{j+1} = K_j -(1/L) \nabla f(K_j)$ will remain in $\ca H$ even if $K_j \in \ca H$. We next observe $1/L$ is indeed a working choice due to the coerciveness of $f(K)$.

The proofs of the following two observations are analogous to their discrete-time counterparts in ~\cite{bu2019lqr_calculus}, and as such, they will be omitted.
\begin{lemma}
  \label{lemma:sequence_bounded}
  For $K_0 \in \ca H$ with $L = \sup_{K \in S_{f(K_0)}}\|\nabla^2 f(K)\|$, the sequence $\{K_j\}$ generated by the scheme~\eqref{eq:gradient_descent} remains in $\ca H$.
\end{lemma}
With gradient dominance property, we can show the following.
\begin{lemma}
  \label{lemma:global_convergence}
  For scheme~\eqref{eq:gradient_descent}, $f(K_n)$ (respectively $K_n$) converges to $f(K_*)$ (respectively $K_*$) at a linear rate, i.e.,
  \begin{align*}
    f(K_n)-f(K_*) \le q^n (f(K_0)-f(K_*)) \\
    \|K_n - K_*\|_F \le c_1 q^n \|K_0 - K_*\|,
  \end{align*}
  where $q \in (0,1)$ and $c_1 > 0$ are constants determined by system parameters $A, B, Q, R$ and initial condition $K_0$.
\end{lemma}
\subsection{Towards choosing a stepsize}
In practice, we would like to choose a stepsize that is determined by system parameters and initial condition. As we have pointed, this is equivalent to estimate an upper bound of the spectral norm of the Hessian $\nabla^2 f(K)$ over the sublevel set $S_{f(K_0)}$. We shall denote the bound of the operator norm of the Hessian $\nabla^2 f(K)$ on the sublevel set by $L$, namely
\begin{align*}
  L = \max_{K \in S_{f(K_0)}}\|\nabla^2 f(K)\| = \max_{K \in S_{f(K_0)}} \sup_{\|E\|_F=1} |\nabla^2 f(K)[E, E]|.
\end{align*}
By triangle inequality,
\begin{align}
\label{eq:hessian_bound}
  \|\nabla^2  f(K)[E, E]\| &\le \sup_{\|E\|_F=1} 2\langle E, REY\rangle \nonumber \\
  &\quad +4\sup_{\|E\|_F=1}|\langle E, B^{\top} X'(K)[E]Y\rangle| \nonumber \\ 
&\le 2\sup_{\|E\|_F=1} \|E^{\top}RE\|_2\Tr(Y) \nonumber \\
&\quad +4\sup_{\|E\|_F=1} \|E^{\top}B^{\top}X'(K)[E]\|_2 \Tr(Y), 
\end{align}
where the last inequality follows from Theorem $2$ in~\cite{mori1988comments}. \par
In what follows, we shall estimate each term in the expression~\eqref{eq:hessian_bound}. We first estimate an upper bound of $Y$ on $S_{f(K_0)}$. Recall $Y$ is the solution to the matrix equation $A_K Y + YA_K^{\top} + X^0 = 0$ and can be written $Y = \int_0^\infty e^{A_K t} X^0 e^{A_K^{\top}t} dt$.
\begin{proposition}
  \label{prop:bound_trace_Y}
  If $K \in S_{f(K_0)}$, then $\Tr(Y) \le f(K_0) \lambda_{\max}(X^0)/(\lambda_{\min}(Q) \lambda_{\min}(X^0))$.
\end{proposition}
\begin{proof}
  We first observe that,
\begin{align*}
  \Tr(Y) &= \Tr\left( \int_0^{\infty} e^{A_K t} X^0 e^{A_K^{\top}t} dt\right) \\
         &= \Tr \left( \int_0^{\infty} e^{A_K^{\top}t} e^{A_K t} dt X^0\right) \\
  &\le \lambda_{\max}(X^0) \Tr \left(\int_0^{\infty} e^{A_K^{\top}t} e^{A_K t} dt \right).
\end{align*}
Putting $Z \coloneqq \int_0^{\infty} e^{A_K^{\top}t} e^{A_K t} dt$ and note $Z$ is the solution to the continuous Lyapunov matrix equation
\begin{align*}
  A_K^{\top} Z + ZA_K + I = 0.
\end{align*}
By Proposition~\ref{prop:linalg_facts},
\begin{align*}
  Z \preceq \frac{1}{\lambda_{\min}(Q)} X,
\end{align*}
since $X^0 \preceq  Q/(\lambda_{\min}(Q)) $. Hence, 
\begin{align}
  \label{eq:bound_trace_Y_eq1}
  \Tr(Y) &\le \frac{\Tr(X) \lambda_{\max}(X^0)}{\lambda_{\min}(Q)} \\
  &\le \frac{\Tr(XX^0) \lambda_{\max}(X^0)}{\lambda_{\min}(X^0) \lambda_{\min}(Q)} \nonumber\\
  &\le \frac{f(K_0) \lambda_{\max}(X^0)}{\lambda_{\min}(X^0)\lambda_{\min}(Q)}.  \nonumber
\end{align}
\end{proof}
We next bound the spectral norm $X'(K)[E]$ on $S_{f(K_0)}$.
\begin{proposition}
  \label{prop:bound_x_prime}
  If $K \in S_{f(K_0)}$ and $\|E\|_F =1$,
\begin{align*}
  \|X'(K)[E]\|_2 \le a f(K_0),
\end{align*}
where $a \in \bb R_+$ is a scalar for which
\begin{align*}
  a Q \succeq \frac{f(K_0)^2}{\lambda_{\min}(X^0)} I + \lambda_{\max}(B^{\top} B) I + \lambda_{\max}(R) I,
\end{align*}
and $a \ge 1$.
\end{proposition}
\begin{proof}
  We note
\begin{align}
  \nonumber
  A_K^{\top} X'(K)[E]&  + X'(K)[E] A_K \\
&= E^{\top}B^{\top} X + X BE  - E^{\top} R K - K^{\top} R E \nonumber \\
  &\preceq X^{\top}X + E^{\top}B^{\top} B E + E^{\top} R E + K^{\top} R K \nonumber \\
  &\preceq X^{\top}X + (\lambda_{\max}(B^{\top}B) + \lambda_{\max}(R))I + K^{\top}R K \label{eq:bound_x_prime_eq1} \\
                     &\preceq \frac{f(K_0)^2}{\lambda_{\min}(X^0)} I + \lambda_{\max}(B^{\top} B) I + \lambda_{\max}(R) I \nonumber \\
  &\quad + K^{\top} R K \eqqcolon V. \nonumber
\end{align}
Choosing $a \in \bb R$ such that $aQ \succeq V-K^{\top} R K$\footnote{If $a \ge 1$, then $a$ is only determined by initial condition and systems parameters; otherwise, we might take $a' =\max(a, 1)$.}, then
\begin{align*}
  X'(K)[E] \preceq a X \preceq a \frac{f(K_0)}{\lambda_{\min}(X^0)} I.
\end{align*}
Reversing all the inequalities, we conclude
\begin{align*}
  \|X'(K)[E]\|_2 \le a \frac{f(K_0)}{\lambda_{\min}(X^0)}.
\end{align*}
\end{proof}
Combing all the bounds we have developed, we have an upper bound of $L$.
\begin{lemma}
  \label{lemma:gradient_lipschitz}
  On $S_{f(K_0)}$, the Lipschitz constant $L$ of the gradient mapping is bounded by
  \begin{align*}
    L \le\left(2 \lambda_{\max}(R) + 2\|B\|_2 a \frac{f(K_0)}P\lambda_{\min}(X^0)\right) \frac{f(K_0) \lambda_{\max}(X^0)}{\lambda_{\min}(X^0) \lambda_{\min}(Q)}.
  \end{align*}
\end{lemma}
\begin{proof}
  It suffices to show the Hessian is bounded by the desired quantity on the sublevel set. By equation~\eqref{eq:hessian_bound}, we have 
  \begin{align*}
    L &\le 2 \lambda_{\max}(R) \Tr(Y) + 2 \|B\|_2 \|X'(K)[E]\|_2 \Tr(Y).
  \end{align*}By Proposition~\ref{prop:bound_trace_Y} and~\ref{prop:bound_x_prime}, we have
  \begin{align*}
    L \le \left(2 \lambda_{\max}(R) + 2\|B\|_2 a \frac{f(K_0)}P\lambda_{\min}(X^0)\right) \frac{f(K_0) \lambda_{\max}(X^0)}{\lambda_{\min}(X^0) \lambda_{\min}(Q)}.
  \end{align*}
\end{proof}
\begin{remark}{\bf Adaptive Stepsize}
We may choose a larger stepsize in each iteration. As such, the resulting stepsizes would be different among iterations. In iteration $n$, we observe we can take the bound of $\Tr(Y_n)$ as inequality~\eqref{eq:bound_trace_Y_eq1} in Proposition~\ref{prop:bound_trace_Y}; for $X'(K_n)[E]$, we note in inequality~\eqref{eq:bound_x_prime_eq1} we may choose $a_n (Q + K_n^{\top} R K_n) \succeq X_n^{\top} X_n + (\|B\|_2^2 + \|R\|_2)I + K_n^{\top}RK_n$. Then in iteration $n$, the Lipschitz bound could be
\begin{align*}
  L_n &\le \left(2 \|R\|_2 + 2\|B\|_2 (\|X_n\|_2^2 + \|R\|_2 + \|B\|_2^2 \right.\\
  &\quad\left. + \|K_n^{\top} R K_n\|_2)\right) \frac{\Tr(X_n)\lambda_{\max}(X^0)}{\lambda_{\min}(Q)} 
\end{align*}
$L_n$ will be dependent on the information of current iteration, namely $K_n$. But we note $L_n$ will be upper bounded by $L$ in Lemma~\ref{lemma:gradient_lipschitz}, so that the stepsizes will be lower bounded by $1/L$.
  \end{remark}

\section{Numerical Results}
\label{sec:numerical_results}
\subsection{Exponential Stability of Gradient Flow $\dot{K}_t = -\nabla f(K_t)$}
We first demonstrate the exponential stability of the gradient system given by
\begin{align*}
  \dot{K}_t = -\nabla f(K_t).
\end{align*}
We choose a gradient flow over a dynamical system modeled over a path graph on $20$ nodes in which each node has its own input. 
The system is given by,
\begin{align*}
  \dot{x}(t) = A x(t) + Bu(t),
\end{align*}
where $A = M - 2I$ with $M$ being the Metropolis-Hastings weighting matrix~\cite{xiao2004fast}\footnote{We subtract the diagonal entries by $2$ to make $A$ Hurwitz stable.} of the path graph and $B=I$, guaranteeing the controllability of the pair $(A, B)$. We set the cost matrices $Q$ and $R$ to the identity matrix. The initial gain matrix is $K_0 = 0$, which belongs to the set of stabilizing feedback gains, based on the spectral properties of the Metropolis-Hastings weighting matrix.
\begin{figure}[ht]
      \begin{center}
     \includegraphics[width=0.4\textwidth]{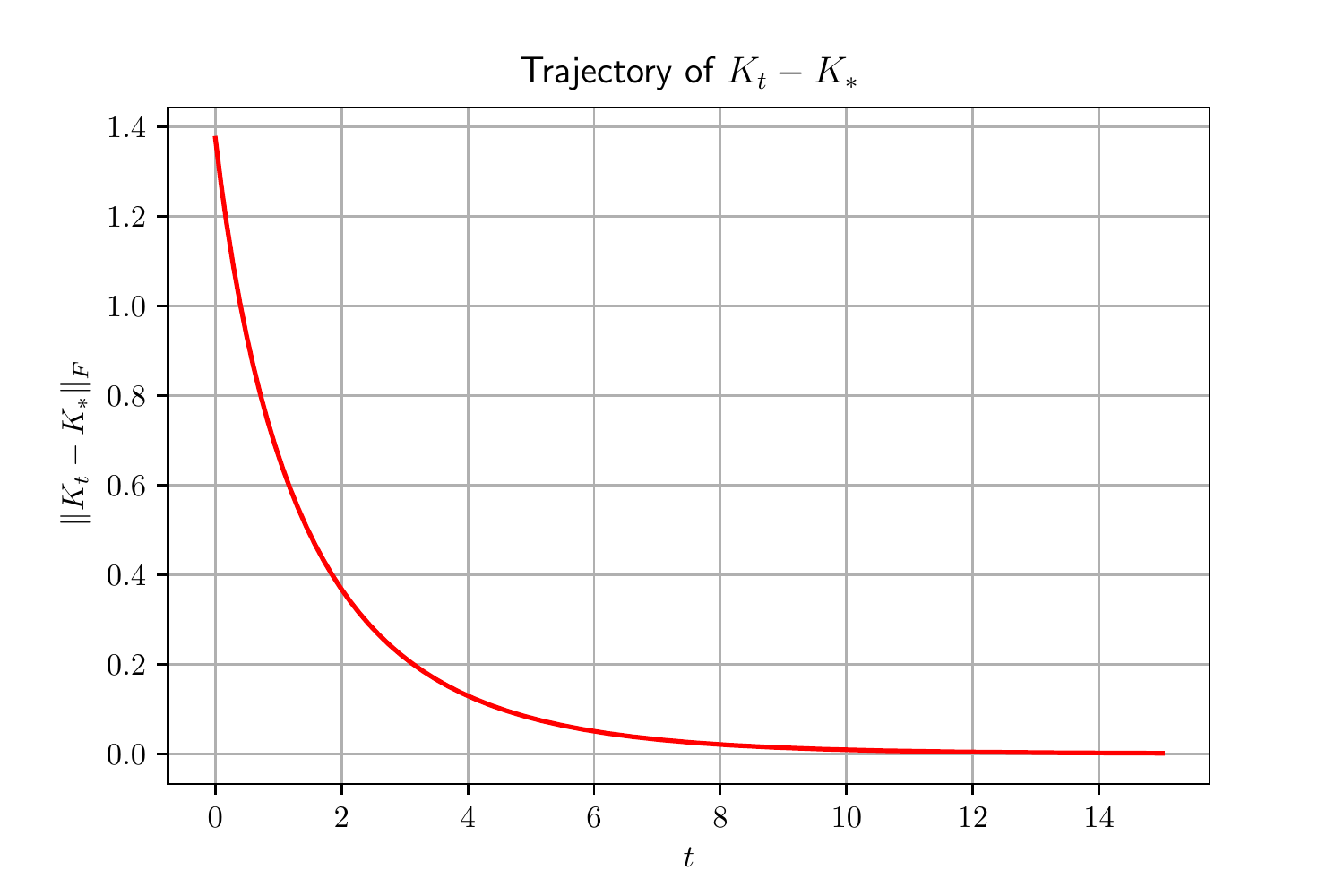}
      \caption{Exponential Stability of Trajectory $K_t$}\label{fig:gf_fig1}
       \end{center}
    \end{figure}
    \begin{figure}[ht]
      \begin{center}
        \includegraphics[width=0.4\textwidth]{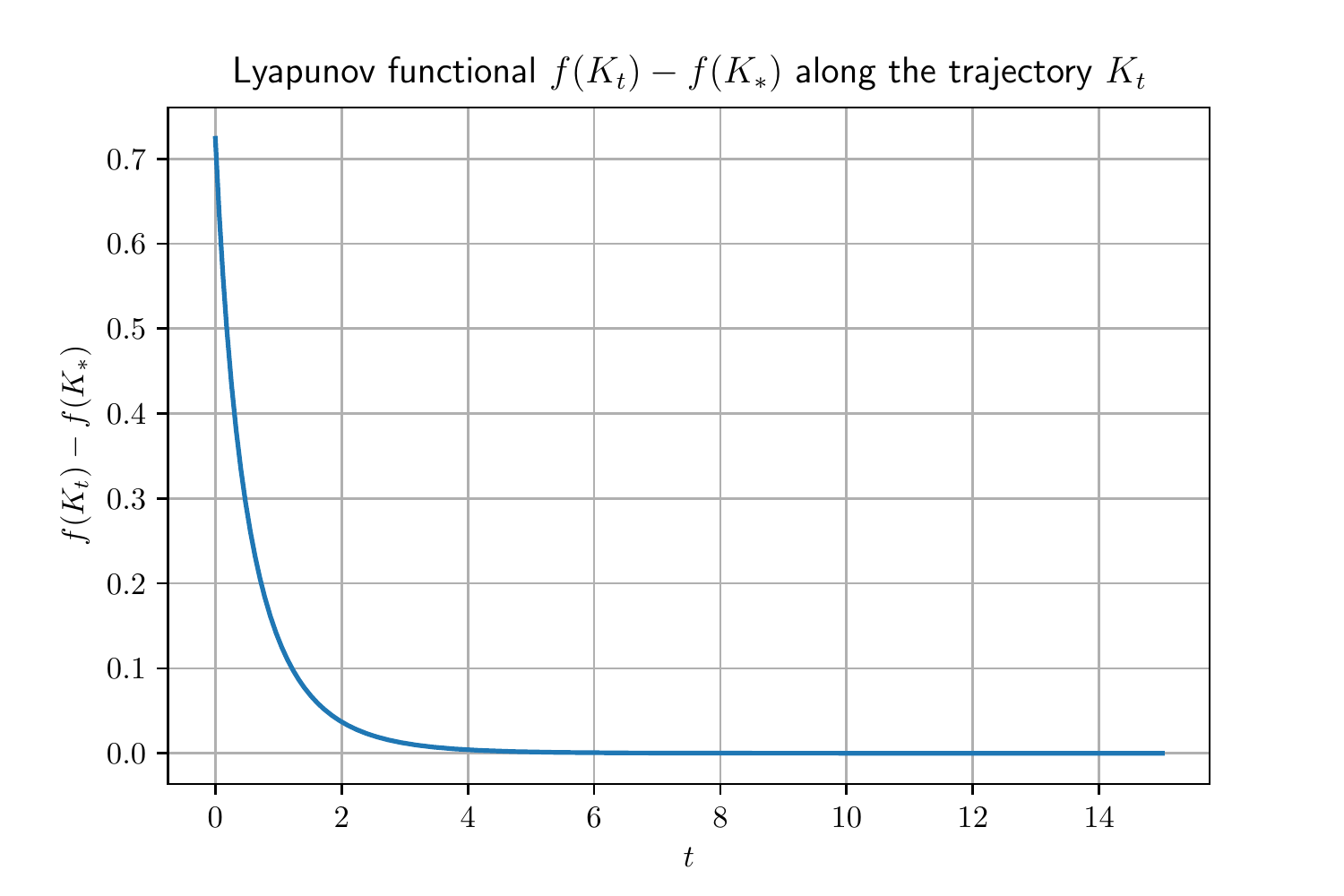}        
        \caption{Exponential Decay of the Lyapunov Functional}\label{fig:df_fig2}
      \end{center}
    \end{figure}
\subsection{Gradient Descent for LQR control}
We demonstrate the proposed discretization procedure for a dynamical system with the same setup above, namely, the dynamical system is modeled over a path graph.\par
As shown in Lemma~\ref{lemma:global_convergence}, Figure~\ref{fig:fig1} demonstrates that the sequence of feedback gains is stabilizing and converges to the global optimal feedback controller. Moreover, Figure~\ref{fig:fig2} shows that the cost function $f(K)$ converges to $f(K^*)$ at a linear rate.
\begin{figure}[ht]
     \begin{center}
     \includegraphics[width=0.43\textwidth]{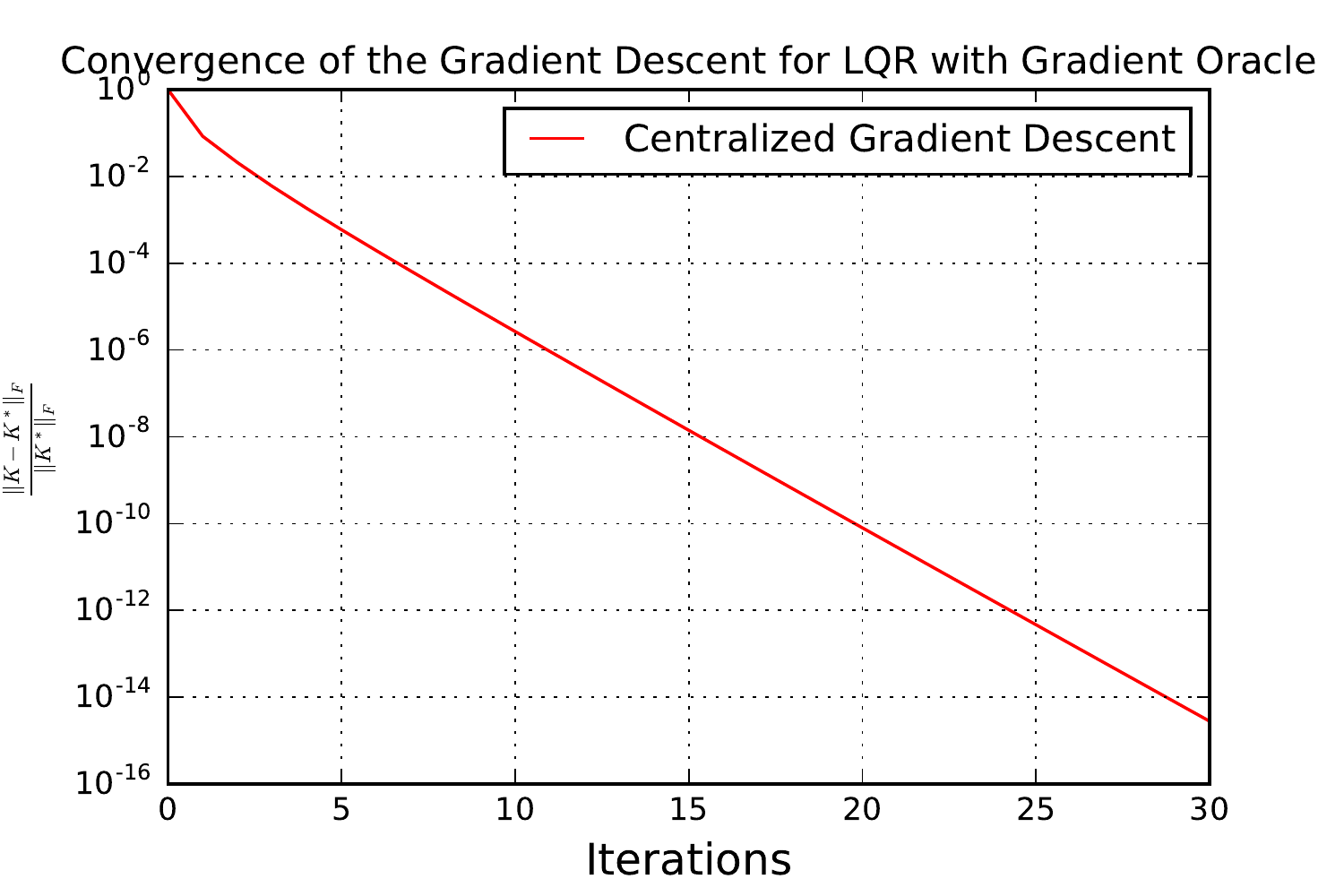}
      \caption{Convergence of the relative error for the feedback gain under gradient descent}\label{fig:fig1}
      \end{center}
    \end{figure}
    \begin{figure}[ht]
     \begin{center}
        \includegraphics[width=0.43\textwidth]{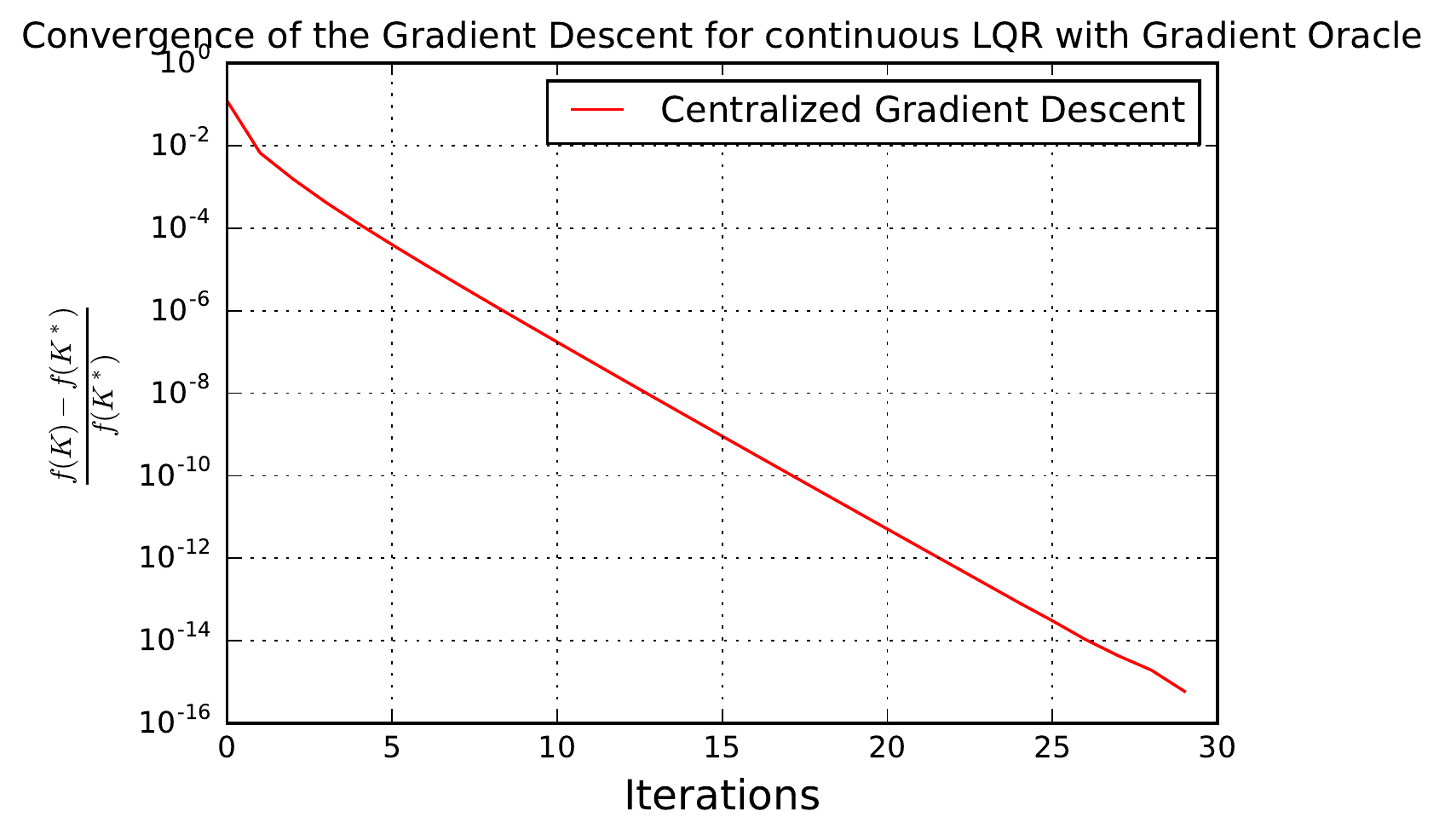}        
        \caption{Convergence of the relative error for the LQR cost function under gradient descent}\label{fig:fig2}
      \end{center}
    \end{figure}
    \subsection{Projected Gradient Descent for LQR control}
    We now demonstrate Projected Gradient Descent for a dynamical system modeled over a $(10, 10)$-lollipop graph\footnote{A lollipop graph consists a complete graph of $10$ nodes and a path graph of $10$ nodes.}. The system parameter $A = M-2I$ is again the Metropolis-Hastings weights matrix subtracted by $2I$ and $B=I$. The initial gain matrix is $K_0 = 0$. In each iteration, the gain matrix is updated by
    \begin{align*}
      K_{n} = P_{\ca U} (K_{n-1} - t \nabla f(K_{n-1})),
    \end{align*}
    where the projection is equivalent to zeroing out the entries that do not correspond to edges of the graph. As shown in Lemma~\ref{lemma:sublinear_pgd}, Figure~\ref{fig:fig1_pgd} demonstrates that the sequence of feedback gains is stabilizing and converges to a first-order stationary point. Moreover, Figure~\ref{fig:fig2_pgd} shows that the cost function $f(K)$ converges. 
\begin{figure}[ht]
     \begin{center}
     \includegraphics[width=0.43\textwidth]{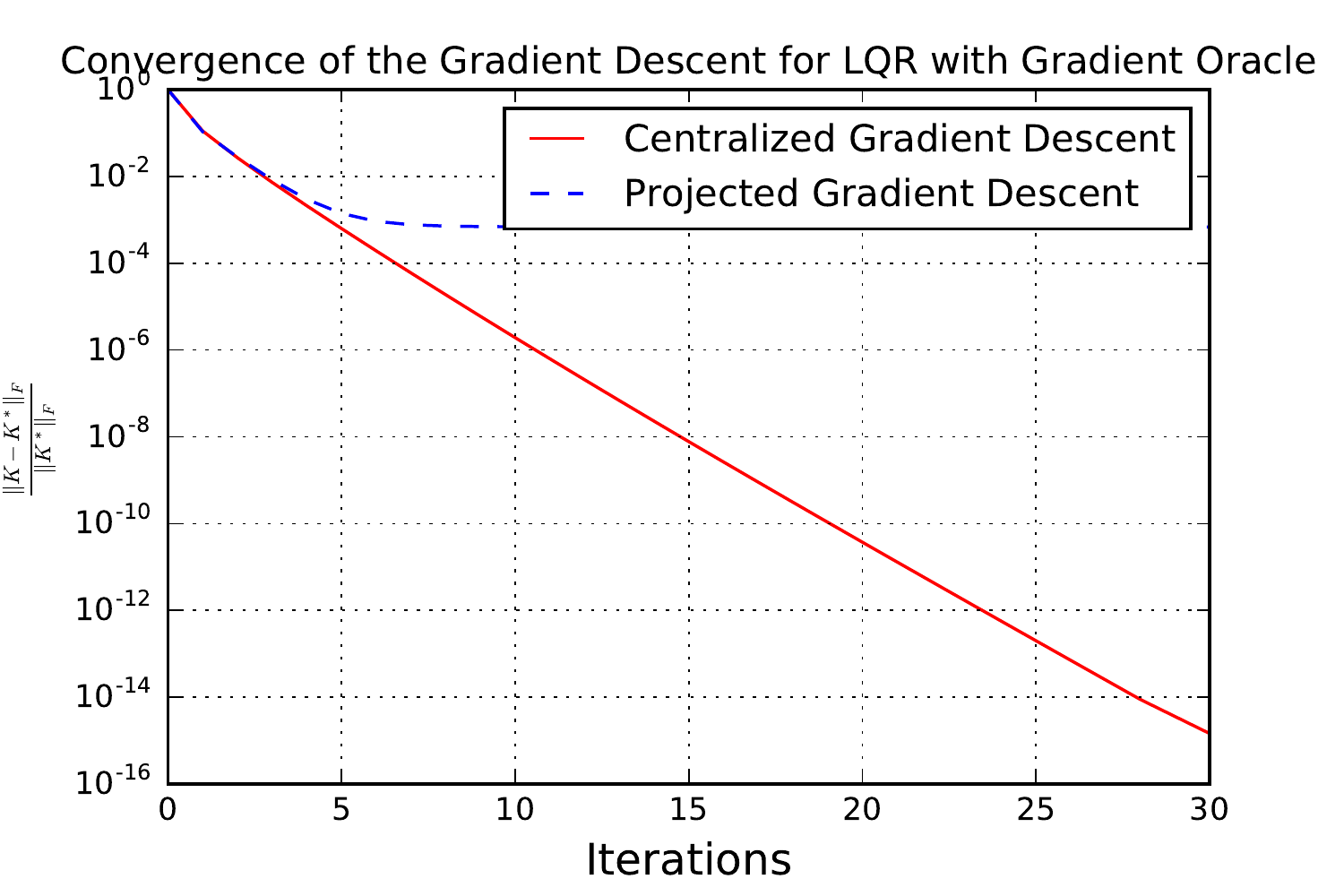}
      \caption{Convergence of the relative error for the feedback gain under centralized gradient descent (red) and projected gradient descent (blue) on a lollipop graph.}\label{fig:fig1_pgd}
      \end{center}
    \end{figure}
    \begin{figure}[ht]
     \begin{center}
        \includegraphics[width=0.43\textwidth]{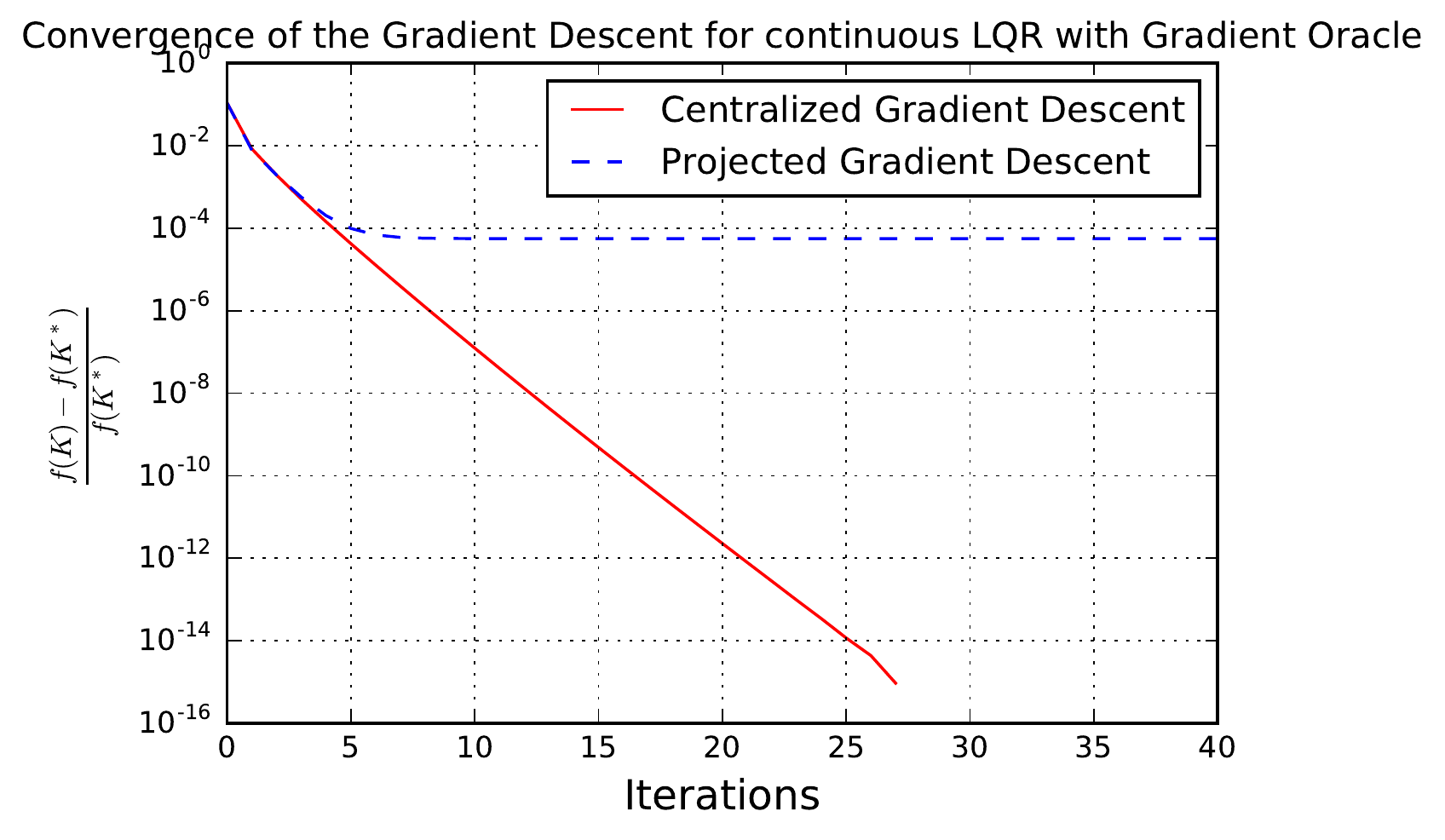}        
        \caption{Convergence of the relative error for the LQR cost function under centralized gradient descent (red) and projected gradient descent (blue) on a lollipop graph.}\label{fig:fig2_pgd}
      \end{center}
    \end{figure}
\section{Concluding Remarks}
\label{sec:discussion}
The direct policy update for LQR as presented in this work has been inspired by recent success
of model-free approaches to optimal control problems that are data-driven, yet can enjoy
certain convergence properties. The work has been influenced by recent contributions
on direct policy updates for discrete-time LQR, and has aimed to highlight some of the differences as well
similarities between the continuous and discrete time settings. Along the way, we have
clarified some of the analytical intricacies of the LQR formulation, and proposed 
three classes of gradient flows and their discretization for globally solving the LQR.
The utility of direct policy updates for structured control synthesis has also been discussed.

\section{Acknowledgements}
The authors acknowledge their discussions with Maryam Fazel, Sham Kakade, and Rong Ge, exploring connections between control theory and learning. This research was supported by DARPA Lagrange Grant FA8650-18-2-7836.
\begin{appendices}
\section{Bounding $\|Y(\theta)\|_2$ and $\|Y'(\theta)\|_2$ in Lemma~\ref{lemma:gd_function_decrease}}
\label{appendix:bound_Y_prime}
Let us first bound $\|Y(\theta)\|_2$.
\begin{proposition}
  Over the sublevel set $S_{f(K)}$, we have for any $K' \in S_{f(K)}$
  \begin{align*}
    \|Y(K')\| \le \frac{f(K)}{\lambda_1(Q)}.
    \end{align*}
  \end{proposition}
  \begin{proof}
    Note $\|Y(K') \| \le \Tr(Y')$. But
    \begin{align*}
      \Tr(Y') &= \Tr \left(\int_0^{\infty} e^{A_{K'}t} {\bf \Sigma} e^{A_{K'}^{\top}t} dt \right) = \Tr\left({\bf \Sigma} \int_0^{\infty} e^{A_{K'}^{\top} t} e^{A_{K'}t} dt\right) \eqqcolon \Tr({\bf \Sigma} Z).
      \end{align*}
      Observe $Z$ solves the equation $A_{K'}^{\top} Z + Z A_{K'} + I = 0$; as $I \preceq Q/\lambda_1(Q)$, $Z \preceq X_{K'}$ where $X_{K'}$ solves the Lyapunov matrix equation $A_{K'}^{\top} X_{K'} + X_{K'} A_{K'} + K'^{\top} R K' + Q =0$. So
      \begin{align*}
        \Tr(Y') \le Tr({\bf \Sigma} Z) \le \frac{\Tr({\bf \Sigma} X')} {\lambda_1(Q)} = \frac{f(K')}{\lambda_1(Q)} \le \frac{f(K)}{\lambda_1(Q)}.
        \end{align*}
   \end{proof}
\begin{proposition}
  Suppose that $K_{\theta} = K - 2 \theta MY$. Then for any $\theta$ such that $f(K_{\theta}) \le f(K)$, we have
\begin{align*}
  \|Y'(\theta)\|_2 \le \frac{4f(K)\|BMY\|_2}{\lambda_1(Q)}.
\end{align*}
\end{proposition}
\begin{proof}
  Putting $A_{\theta}= A-BK + 2 \theta BMY$, $Y_{\theta}$ solves
\begin{align*}
  0 = A_{\theta} Y_{\theta} + Y_{\theta} A_{\theta}^{\top} + {\bf \Sigma}.
\end{align*}
Therefore, the derivative $Y_{\theta}'$ solves
\begin{align*}
   -A_{\theta} Y_{\theta}' - Y_{\theta}' A_{\theta}^{\top} &= 2BMY Y_{\theta} + Y_{\theta} (2BMY)^{\top} \\
  &\preceq \tau 4BMY (BMY)^{\top} + \frac{1}{\tau}  Y_{\theta}  \\
                                                           &\preceq (4 \tau\|BMY\|_2^2  + \frac{1}{\tau} \lambda_n(Y_{\theta})) I.
\end{align*}
This upper bound can be minimized by taking $\tau = \sqrt{\frac{ 1}{4\|BMY\|_2^2}}$. It thus follows that,
\begin{align*}
  -A_{\theta} Y_{\theta}' - Y_{\theta} A_{\theta}^{\top} &\le 4 \|BMY\|_2 \lambda_n(Y_{\theta})I.
  \end{align*}
  By Proposition~\ref{prop:linalg_facts}, we have
  \begin{align*}
    Y_{\theta}' \preceq \frac{4\|BMY\|_2 \lambda_n(Y_{\theta})}{\lambda_1({\bf \Sigma})} Y.
    \end{align*}
  Similarly we can prove
\begin{align*}
  Y_{\theta}' \succeq - \frac{4\|BMY\|_2\lambda_n(Y_{\theta})}{\lambda_1({\bf \Sigma})} Y.
  \end{align*}
  So we have
  \begin{align*}
    \|Y_\theta'\|_2 \le \frac{4\|BMY\|_2 f(K)}{\lambda_1(Q)\lambda_1({\bf \Sigma})}.
    \end{align*}
\end{proof}
\section{Lower Bounding Stepsize $\eta_j$ in Gradient Descent}
\label{appendix:upperbound_stepsize}
We first bound the constant $b_j, c_j$ defined in Lemma~\ref{lemma:gd_function_decrease}.
\begin{proposition}
  \label{prop:appendix_bound_b_c}
  Over the sublevel set $S_{f(K_0)}$,
  \begin{align*}
    b_j &\le \frac{ \lambda_n(R) f(K_0) + \|B\|^2 \frac{(f(K_0))^2}{\lambda_1({\bf \Sigma})} + 2 \frac{(f(K_0))^2 \|B\|_2 \sqrt{(\lambda_n(R)+\|B\|_2^2 \frac{f(K_0)}{\lambda_1({\bf \Sigma})}) f(K_0)}}{\lambda_1(Q)}}{\lambda_1(Q)}, \\
    c_j &\le \frac{ 
          2(\lambda_n(R) + \|B\|^2 \frac{f(K_0)}{\lambda_1({\bf \Sigma})})
          \|B\|_2 f(K_0)  \sqrt{(\lambda_n(R)+\|B\|_2^2 \frac{f(K_0)}{\lambda_1({\bf \Sigma})}}
          \frac{ f(K_0)}{\lambda_1(Q)}}
{\lambda_1(Q)}.
  \end{align*}
\end{proposition}
\begin{proof}
  To derive an upper bound, we need to upper bound $\lambda_n(R+B^{\top} X_j B), \|Y_j\|_2, \|M_{K_j}\|_2$. Note by virtually the same argument in Proposition~\ref{prop:gradient_dominant_bound}, we can upper bounded $\|Y_j\|_2$ by
\begin{align*}
  \|Y_j\|_2 \le \Tr(Y_j) \le \frac{f(K_0)}{\lambda_1(Q)}.
\end{align*}
By the proof in Theorem~\ref{thrm:ngd_convergence}, we can upper bound $\|M_{K_j}\|_2$ by
\begin{align*}
  \|M_{K_j}\|_2^2 \le \Tr(M_{K_j}^{\top} M_{K_j}) \le \lambda_n(R) f(K_0).
\end{align*}
  The desired bounds can then be acquired by combining these bounds and triangular inequality.
\end{proof}
\begin{proposition}
  \label{prop:stepsize_bound}
  The stepsize in Theorem~\ref{thrm:gd_linear} is lower bounded away from $0$.
\end{proposition}
\begin{proof}
  Putting $d_j = \max\{b_j, c_j\}$, then $d_j \le \delta$ over the sublevel set $S_{f(K_0)}$ where $\delta > 0$ is the maximum of the two upper bounds in Proposition~\ref{prop:appendix_bound_b_c}. Now $\eta_j$ can be seen
\begin{align*}
  \eta_j = \sqrt{\frac{1}{d_j}+\frac{1}{9}} - \frac{1}{3} \ge \sqrt{\frac{1}{\delta} + \frac{1}{9}} -\frac{1}{3}> 0.
\end{align*}
\end{proof}
\end{appendices}
\bibliographystyle{ieeetran}
\bibliography{ref}
\end{document}